\documentclass[preprint]{sig-alternate}

\usepackage[english]{babel}
\usepackage{verbatim}
\usepackage{suffix}
\usepackage{times}
\usepackage{amsmath,amsfonts}
\usepackage{MnSymbol} \usepackage{amssymb}
\usepackage{bbm}
\usepackage{cite}
\usepackage{url}
\usepackage{graphicx}
\usepackage{thm-restate}
\usepackage[usenames,dvipsnames]{xcolor}
\usepackage{soul}
\usepackage{algpseudocode}
\usepackage{algorithm}
\usepackage[justification=centering]{caption}
\usepackage{subcaption,float}
\usepackage{mdwlist}
\usepackage{multirow}
\usepackage{balance}
\usepackage{xspace}
\usepackage{enumitem}
\usepackage{hyperref}
\usepackage[capitalize]{cleveref}

\crefname{principle}{Principle}{Principles} 
\Crefname{principle}{Principle}{Principles} 
\crefname{finding}{Finding}{Findings} 
\Crefname{finding}{Finding}{Findings} 
\crefname{point}{Point}{Points}
\Crefname{point}{Point}{Points}

\crefname{commentcounter}{R\!}{Rs} 
\Crefname{commentcounter}{R\!}{Rs}

\begin{document}
%!TEX root = paper.tex

\newtheorem{theorem}{Theorem}
\newtheorem{lemma}{Lemma}
\newtheorem{definition}{Definition}
\newtheorem{problem}{Problem}
\newtheorem{proposition}{Proposition}
\newtheorem{corollary}{Corollary}
\newcommand{\example}[1]{{\bf Example}: {#1}}
\newtheorem{principle}{Principle}
\newtheorem{finding}{Finding}

\newcommand{\mybullet}{\vspace{1mm}\noindent$\bullet$~~}

\renewcommand{\hl}[1]{#1}
\newcommand{\mh}[1]{[[\emph{\color{blue}MH: #1}]]}

\newcommand{\vect}[1]{\mathbf{#1}}
\newcommand{\E}{\mathbb{E}}

\newcommand{\scale}[1]{\Lone{#1}}
\newcommand{\dpbench}{{\sc DPBench}\xspace}
\newcommand{\dpbenchone}{{\sc DPBench-1D}\xspace}
\newcommand{\dpbenchtwo}{{\sc DPBench-2D}\xspace}
\newcommand{\datagen}{{\sc DataGen}\xspace}

\def\aa{\mathbb{A}}  
\def\bb{\mathbb{B}}  

\def\e{\vect{e}}
\def\x{\vect{x}}
\def\p{\vect{p}}
\def\estx{\vect{\hat x}}
\def\tildex{\vect{\tilde x}}
\def\ests{\vect{\hat s}}
\def\W{\vect{W}}
\def\z{\vect{z}}
\def\Z{\vect{Z}}

\def\db{I}  \def\nbrs{nbrs}
\def\dbdom{\mathbb{Z}_{\geq 0}^n}  \def\allbuckets{\mathcal{B}}
\def\allhist{\mathcal{P}}
\def\bmean{\frac{b_i(\x)}{|b_i|}}  \def\unexp{{expand}}

\def\hcost{pcost}
\def\bcost{bcost}

\def\real{\mathbb{R}}
\def\posreal{\mathbb{R}_{\geq 0}}
\def\sens{\Delta}

\newcommand{\Lone}[1]{\left\Vert #1  \right\Vert_1}
\newcommand{\Ltwo}[1]{\left\Vert #1  \right\Vert_2}
\newcommand{\set}[1]{\{#1\}}   
\def\myvert{\;\vert\;}

\def\B{B}

\newcommand{\stitle}[1]{\vspace{-2mm}\paragraph*{#1}}

\def\param{\rho}
\def\seed{\sigma}

\def\ones{\vect{1}}
\def\zeros{\vect{0}}
\def\plus{{\!+}}
\def\b{\vect{b}}  
\def\a{\vect{a}}
\def\i{\vect{i}}
\def\y{\vect{y}}
\def\q{\vect{q}}  
\def\w{\vect{w}} 
\def\v{\vect{v}}  
\def\estw{\vect{\hat w}}
\def\estq{\vect{\hat q}}
\def\A{\vect{A}}
\def\T{\vect{T}}
\def\Q{\vect{Q}}
\def\M{\vect{M}}
\def\D{\vect{D}}
\def\d{\vect{d}}
\def\P{\vect{P}}
\def\p{\vect{p}}
\def\I{\vect{I}}
\def\V{\vect{V}}
\def\H{\vect{H}}
\def\G{\vect{G}}
\def\R{\vect{R}}
\def\X{\vect{X}}
\def\Y{\vect{Y}}
\def\s{\vect{s}}
\def\tq{\hat{q}}
\def\tW{\hat{W}}
\def\tWW{\mathbf{\tW}}
\def\lambdaB{\vect{\lambda}}
\def\LambdaB{\vect{\Lambda}}

\def\esty{\vect{\hat y}}
\def\m{\vect{m}}

\def\dom{\mathbb{N}^n}
\def\rng{\mathcal{Y}}
\def\obj{y}

\def\uni{\mathcal{U}}
\def\coll{\mathcal{S}}
\def\packing{\mathcal{C}}
\def\Lap{\mbox{Laplace}}

\def\dom{\mathcal{D}}\def\reg{reg}

\def\cover{H}
\def\allparts{\mathcal{P}(\allbuckets)}

\def\DWnospace{\text{DAWA}}
\def\DW{\text{DAWA} }
\def\DWall{\text{DAWA-all} }
\def\DWapprox{\text{DAWA-subset} }

\def\algG{\mathcal{A}}  \def\alg{\mathcal{K}}  \def\LM{\mathcal{L}}	\def\GM{\mathcal{G}}	\def\MM{\mathcal{M}}	\def\hist{{hist}}	
\def\7{\surd}

\def\propPart{\rho}
\newcommand{\uniform}{\text{{\sc Uniform}}\xspace}
\newcommand{\identity}{\text{{\sc Identity}}\xspace}
\newcommand{\privelet}{\text{{\sc Privelet}}\xspace}
\newcommand{\Hier}{\text{{\sc H}}\xspace}
\newcommand{\Hierb}{\text{{\sc H$_b$}}\xspace}
\newcommand{\greedyH}{\text{{\sc Greedy H}}\xspace}
\newcommand{\dawa}{\text{{\sc DAWA}}\xspace}
\newcommand{\mwem}{\text{{\sc MWEM}}\xspace}
\newcommand{\mwemv}{\text{{\sc MWEM$^*$}}\xspace}
\newcommand{\efpa}{\text{{\sc EFPA}}\xspace}
\newcommand{\php}{\text{{\sc PHP}}\xspace}
\newcommand{\ahp}{\text{{\sc AHP}}\xspace}
\newcommand{\ahpv}{\text{{\sc AHP$^*$}}\xspace}
\newcommand{\structurefirst}{\text{{\sc SF}}\xspace}
\newcommand{\dpcube}{\text{{\sc DPCube}}\xspace}
\newcommand{\quadtree}{\text{{\sc QuadTree}}\xspace}
\newcommand{\hybridtree}{\text{{\sc HybridTree}}\xspace}
\newcommand{\UG}{\text{{\sc UGrid}}\xspace}
\newcommand{\AG}{\text{{\sc AGrid}}\xspace}

\newcommand{\patent}{\text{{\sc Patent}}\xspace}
\newcommand{\hepph}{\text{{\sc HepPh}}\xspace}
\newcommand{\income}{\text{{\sc Income}}\xspace}
\newcommand{\adult}{\text{{\sc Adult}}\xspace}
\newcommand{\searchlogs}{\text{{\sc Search}}\xspace}
\newcommand{\nettrace}{\text{{\sc Trace}}\xspace}
\newcommand{\medcost}{\text{{\sc Medcost}}\xspace}
\newcommand{\bidsfj}{\text{{\sc Bids-FJ}}\xspace}
\newcommand{\bidsfm}{\text{{\sc Bids-FM}}\xspace}
\newcommand{\bidsall}{\text{{\sc Bids-All}}\xspace}
\newcommand{\mdsalary}{\text{{\sc MD-Sal}}\xspace}
\newcommand{\mdsalaryfa}{\text{{\sc MD-Sal-FA}}\xspace}

\newcommand{\lcreqfa}{\text{{\sc LC-Req-F1}}\xspace}
\newcommand{\lcreqfb}{\text{{\sc LC-Req-F2}}\xspace}
\newcommand{\lcreqall}{\text{{\sc LC-Req-All}}\xspace}
\newcommand{\lcdtirfa}{\text{{\sc LC-DTIR-F1}}\xspace}
\newcommand{\lcdtirfb}{\text{{\sc LC-DTIR-F2}}\xspace}
\newcommand{\lcdtirall}{\text{{\sc LC-DTIR-All}}\xspace}
\newcommand{\TDadult}{\text{{\sc Adult-2D}}\xspace}
\newcommand{\TDloan}{\text{{\sc LC-2D}}\xspace}
\newcommand{\TDtwitter}{\text{{\sc TWITTER-100-100}}\xspace}
\newcommand{\TDcheckin}{\text{{\sc Gowalla}}\xspace}
\newcommand{\TDstoke}{\text{{\sc Stroke}}\xspace}
\newcommand{\TDmdsalary}{\text{{\sc MD-Sal-2D}}\xspace}
\newcommand{\TDbeijingtaxis}{\text{{\sc BJ-Cabs-S}}\xspace}
\newcommand{\TDbeijingtaxie}{\text{{\sc BJ-Cabs-E}}\xspace}
\newcommand{\TDcabspottings}{\text{{\sc SF-Cabs-S}}\xspace}
\newcommand{\TDcabspottinge}{\text{{\sc SF-Cabs-E}}\xspace}
\newcommand{\TDlbl}{\text{{\sc LBL}}\xspace}

\newcommand{\randomworkload}{\text{{\em Random}}\xspace}
\newcommand{\randomsmall}{\text{{\em Random Small}}\xspace}
\newcommand{\prefixworkload}{\text{{\em Prefix}}\xspace}

\toappear{} 

\newcommand\debullet[2]{#2}

\title{Principled Evaluation of Differentially Private Algorithms using DPBench}

\numberofauthors{3}
\author{
		Michael Hay\textsuperscript{\textasteriskcentered}, 
	Ashwin Machanavajjhala\textsuperscript{\textasteriskcentered\textasteriskcentered}, 
	Gerome Miklau\textsuperscript{\textdagger}, 
	Yan Chen\textsuperscript{\textasteriskcentered\textasteriskcentered}, 
	Dan Zhang\textsuperscript{\textdagger} \and
		\alignauthor
	       \affaddr{
	       \makebox[0pt][r]{\textsuperscript{\textasteriskcentered}}
	       Colgate University}\\
	       \affaddr{Department of Computer Science}\\
	       \affaddr{mhay@colgate.edu} 
		\alignauthor
	       \affaddr{
	       \makebox[0pt][r]{\textsuperscript{\textasteriskcentered\textasteriskcentered}}
	       Duke University}\\
	       \affaddr{Department of Computer Science}\\
	       \affaddr{\{ashwin,yanchen\}@cs.duke.edu} 
	\alignauthor
	       \affaddr{
	       \makebox[0pt][r]{\textsuperscript{\textdagger}}
	       University of Massachusetts Amherst\\}
		   \affaddr{School of Computer Science}\\
	       \affaddr{ \{miklau,dzhang\}@cs.umass.edu} 
}

\maketitle
\pagestyle{plain}
\thispagestyle{empty}
\pagenumbering{arabic}

\begin{abstract}
%!TEX root = paper.tex

Differential privacy has become the dominant standard in the research community for strong privacy protection.  There has been a flood of research into query answering algorithms that meet this standard.  Algorithms are becoming increasingly complex, and in particular, the performance of many emerging algorithms is {\em data dependent}, meaning the distribution of the noise added to query answers may change depending on the input data.  Theoretical analysis typically only considers the worst case, making empirical study of average case performance increasingly important. 

In this paper we propose a set of evaluation principles which we argue are essential for sound evaluation.  Based on these principles we propose \dpbench, a novel evaluation framework for  standardized evaluation of privacy algorithms.  We then apply our benchmark to evaluate algorithms for answering 1- and 2-dimensional range queries.
The result is a thorough empirical study of \hl{15} published algorithms on a total of \hl{27} datasets that offers
 new insights into algorithm behavior---in particular the influence of dataset scale and shape---and a more complete characterization of the state of the art.  Our methodology is able to resolve inconsistencies in prior empirical studies and place algorithm performance in context through comparison to simple baselines.
Finally, we pose open research questions 
which we hope will guide future algorithm design.

\end{abstract}

%!TEX root = paper.tex
\section{Introduction}\label{sec:intro}

Privacy is a major obstacle to deriving important insights from collections of sensitive records donated by individuals.  Differential privacy \cite{dwork2006calibrating, dwork2011a-firm, Dwork08Differential, Dwork14Algorithmic} has emerged as an important standard for protection of individuals' sensitive information. 
Informally, differential privacy is a property of the analysis algorithm which guarantees that the output the analyst receives is statistically indistinguishable (governed by a privacy parameter $\epsilon$) from the output the analyst would have received if any one individual had opted out of the collection.  
Its general acceptance by researchers has led to a flood of work across the database, data mining, theory, machine learning, programming languages, security, and statistics communities.

Most differentially private algorithms work by introducing noise into query answers. Finding algorithms that satisfy $\epsilon$-differential privacy and introduce the least possible error for a given analysis task is a major ongoing challenge both in research and practice. Standard techniques for satisfying differential privacy that are broadly-applicable (e.g. the Laplace and exponential mechanisms \cite{Dwork14Algorithmic}) often offer sub-optimal error rates. Much recent work that deem differential privacy as impractical for real world data (e.g., \cite{Hu2015Telco}) use only these standard techniques. 

These limitations have led to a slew of work in new differentially private algorithms that reduce the achievable error rates. Take the example task of privately answering 1- and 2-dimensional range queries on a dataset (the primary focus of this paper). Proposed techniques for this task include answering range queries using noisy hierarchies of equi-width histograms \cite{Cormode11Differentially,Qardaji13Understanding,hay2010boosting, Li:2010Optimizing-Linear},  noisy counts on a coarsened domain \cite{xu2013differential, zhangtowards, Li14Data-,xiao2014dpcube,qardaji2013differentially}, or by reconstructing perturbed wavelet \cite{xiao2010differential} or Fourier \cite{Acs2012compression} coefficients,  or based on a synthetic dataset built using multiplicative weights updates \cite{hardt2012a-simple}. These recent innovations  reduce error at a fixed privacy level $\epsilon$ by many orders of magnitude for certain datasets, and can have a large impact on the success of practical differentially private applications. 

However, the current state-of-the-art poses a new challenge of algorithm selection. Consider a data owner (say from the US Census Bureau) who would like to use differentially private algorithms to release a 1- or 2- dimensional histogram over their data. She sees a wide variety of algorithms in the published literature, each demonstrating settings or contexts in which they have advantageous theoretical and empirical properties. Unlike in other fields (e.g. data mining), the data owner can not run all the algorithms on her dataset and choose the algorithm that incurs the least error -- this violates differential privacy, as the choice of the algorithm would leak information about the input dataset.  Hence, the data owner must make this choice using prior theoretical and empirical analyses of these algorithms, and faces the following problems: 
\begin{enumerate}[leftmargin=*]
\item {\em Gaps in Empirical Evaluations:} As algorithms become progressively more complex, their error rates are harder to analyze theoretically, underscoring the importance of good empirical evaluations. For a number of reasons, including chronology of algorithms, lack of benchmark datasets, and space constraints in publications, existing empirical evaluations do not comprehensively evaluate all existing algorithms leaving gaps (and even inconsistencies; see \cref{sec:prior_results}) in our knowledge about algorithm performance. 
\item {\em Understanding Data-Dependence:} A number of recent algorithms are {\em data dependent}; i.e., their error is sensitive to properties of the input. Thus, a data dependent algorithm $A$ may have lower error than another algorithm $B$ on one dataset, but the reverse may be true on another dataset. While a few prior empirical evaluations do evaluate algorithms on diverse datasets, there is little guidance for the data owner on how the error would extrapolate to a new dataset (e.g., one with a much larger number of tuples) or a new experimental setting (e.g., a smaller  $\epsilon$). 
\item {\em Choosing Values for Free Parameters:} Algorithms are often associated with free parameters (other than the privacy parameter $\epsilon$), but their effect on error is incompletely quantified. Published research provides little guidance to the data owner on how to set these parameters, or default values are suboptimal. 
\item {\em Unsound Utility Comparisons:} Even when empirical analyses are comprehensive, the results may not be useful to a practitioner. For instance, the error of a differentially private algorithm is a random variable. However, most empirical analyses only report the mean error and not the variability in error. Moreover, algorithm error is often not compared with that of simple baselines like the Laplace mechanism, which are the first algorithms a practitioner would attempt applying.
\end{enumerate} 

Given these problems, even as research progresses, the practitioner is lost and incapable of deploying the right algorithm. Moreover, researchers are likely to propose new algorithms that improve performance in narrow contexts. In this paper, we attempt to remedy the above problems with a state-of-the-art of empirical evaluation of differentially private algorithms, and help shed light on the algorithm selection problem for 1- and 2-dimensional range query answering. 
We make the following novel contributions: \begin{enumerate}[leftmargin=*]
\item We propose a set of principles that any differentially private evaluation must satisfy (\cref{sec:principles}). Based on these principles we develop \dpbench, a novel methodology for evaluating differentially private algorithms (\cref{sec:benchmark}). Experiments designed using \dpbench help tease out dependence of algorithm error on specific data characteristics like the size of the domain, the dataset's number of tuples (or scale) and its empirical distribution (or shape). \dpbench provides an algorithm to automatically tune free parameters of algorithms, and ensures fair comparisons between algorithms. Finally, by reporting both the mean and variation of error as well as comparing algorithms to baselines the results of \dpbench experiments  help the practitioner get a better handle on the algorithm selection problem.  
\item Using \dpbench, we present a comprehensive empirical study of \hl{15} published algorithms for releasing 1- or 2-dimensional range queries on a total of 27 datasets (\cref{sec:setup}). We evaluate algorithm error on \hl{7,920} experimental settings. Our study presents the following novel insights into algorithm error (\cref{sec:evaluation}): \begin{enumerate}[leftmargin=*]
\item {\em Scale and Data Dependence:} The error incurred by recently proposed data dependent algorithms is heavily influenced by scale (number of tuples). At smaller scales, the best data-dependent algorithms can beat simpler data independent algorithms by an order of magnitude. However, at large scales, many data dependent algorithms start performing worse than data independent algorithms. \item  {\em Baselines:} A majority of the recent algorithms do not even consistently beat the Laplace mechanism, especially at medium and larger scales in both the 1- and 2-D cases. 
\item {\em Understanding Inconsistencies:} Our results also help us resolve and explain inconsistencies from prior work. For instance, a technique that uses multiplicative weights \cite{hardt2012a-simple} was shown to outperform a large class of  data independent mechanisms (like \cite{xiao2010differential}), but in a subsequent paper \cite{Li14Data-}, the reverse was shown to be true. Our results explain this apparent inconsistency by identifying that the former conclusion was drawn using a dataset with small scale while the latter conclusions were made on datasets with larger scales.  
\end{enumerate}
\item  While the \dpbench algorithm for tuning free parameters is quite straightforward, we show that it results in $13\times$ improvement in error for certain algorithms like the multiplicative weights method \cite{hardt2012a-simple} over the original implementations of where parameters are set in a default manner.  
\item We formalize two important theoretical properties, {\em scale-epsilon exchangeability} and {\em consistency}, that can help extrapolate results from an empirical study to other experimental setting not considered by the study.  For an algorithm that is {\em scale-espilon exchangeable}, increasing the scale of the input dataset and increasing epsilon have equivalent effects on the error. Thus, an algorithm's error is roughly the same for all scale and $\epsilon$ pairs with the same product. We prove that most algorithms we consider in this paper satisfy this property. An algorithm that is {\em consistent} has error that tends to 0 as the privacy parameter $\epsilon$ tends to $\infty$. That is, inconsistent algorithms return answers that are biased even in the absence of privacy constraints. While all the data independent algorithms considered in the paper are consistent, we show that some data dependent algorithms are not consistent, and hence a practitioner must be wary about using such algorithms. 
\end{enumerate}

While our analysis is restricted to 1- and 2-dimensional range queries, we believe our results are very useful and impactful as the evaluation principles will extend to evaluation of differentially private algorithms for other tasks. We conclude the paper with comparisons to published results, lessons for practitioners, and a discussion of open questions in \cref{sec:discussion}.

%!TEX root = paper.tex
\section{Preliminaries}\label{sec:background}

In this section we review basic privacy definitions and introduce notation for databases and queries. 
\subsection{Differential Privacy} \label{sec:diffp}

In this paper we are concerned with algorithms that run on a private database and publish their output.  To ensure the released data does not violate privacy, we require the algorithms to obey the standard of differential privacy.

Let $\db$ be a database instance consisting of a single relation.  Let $\nbrs(\db)$ denote the set of databases differing from $I$ in at most one record; i.e., if $\db' \in \nbrs(\db)$, then $|(\db - \db') \cup (\db' - \db)| = 1$.

\begin{definition}[Differential Privacy~\cite{dwork2006calibrating}] \label{def:diffp}
A randomized algorithm $\algG$ is $\epsilon$-differentially private if for any instance $\db$, any $\db' \in \nbrs(\db)$, and any subset of outputs $S \subseteq Range(\algG)$,\[
Pr[ \algG(\db) \in S] \leq \exp(\epsilon) \times Pr[ \algG(\db') \in S]
\]		
\end{definition}
For an individual whose data is represented by a record in $\db$, differential privacy offers an appealing guarantee.  It says that including this individual's record cannot significantly affect the output: it can only make some outputs slightly more (or less) likely -- where ``slightly'' is defined as at most a factor of $e^{\epsilon}$.  If an adversary infers something about the individual based on the output, then the same inference would also be likely to occur {\em even if the individual's data had been removed from the database} prior to running the algorithm.

Many of the algorithms considered in this paper are composed of multiple subroutines, each taking the private data as input.  Provided each subroutine achieves differential privacy, the whole algorithm is differentially private.  More precisely, the sequential execution of $k$ algorithms $\algG_1, \dots, \algG_k$, each satisfying $\epsilon_i$-differential privacy, results in an algorithm that is $\epsilon$-differentially private for $\epsilon = \sum_i \epsilon_i$~\cite{mcsherry2009pinq}.  Hence, we may think of $\epsilon$ as representing an algorithm's {\em privacy budget} which can be allocated across its subroutines.

A commonly used subroutine is the Laplace mechanism, a general purpose algorithm for computing numerical functions on the private data.  It achieves privacy by adding noise to the function's output.
We use $\Lap(\sigma)$ to denote the Laplace probability distribution with mean 0 and scale $\sigma$. 
\begin{definition}[Laplace Mechanism~\cite{dwork2006calibrating}] \label{def:laplace_mechanism}
Let $f(\db)$ denote a function on $\db$ that outputs a vector in $\real^d$.  The Laplace mechanism $\LM$ is defined as $\LM(I) = f(I) + \z$, where $\z$ is a $d$-length vector of random variables such that $z_i \sim \Lap(\sens f / \epsilon)$.  

The constant $\sens f$ is called the {\em sensitivity of $f$} and is the maximum difference in $f$ between any two databases that differ only by a single record,
$
\sens f = \max_{\db, \db' \in \nbrs(\db)} \Lone{ f(\db) - f(\db') }
$.
\end{definition}
The Laplace mechanism can be used to provide noisy counts of records satisfying arbitrary predicates.  For example, suppose $\db$ contains medical records and $f$ reports two counts: the number of male patients with heart disease and the number of female patients with heart disease.  The sensitivity of $f$ is 1: given any database instance $\db$, adding one record to it (to produce neighboring instance $\db'$), could cause at most one of the two counts to increase by exactly 1.  Thus, the Laplace mechanism would add random noise from $\Lap(1 / \epsilon)$ to each count and release the noisy counts.

\subsection{Data Model and Task} \label{sec:data_model}

The database $\db$ is an instance of a single-relation schema $R(\aa)$, with attributes $\aa=\{A_1, A_2, \ldots, A_{\ell}\}$.  Each attribute is discrete, having an ordered domain (continuous attributes can be suitably discretized).  We are interested in answering range queries over this data; range queries support a wide range of data analysis tasks including histograms, marginals, data cubes, etc.  
We consider the following task.  The analyst specifies a subset of target attributes, denoted $\bb \subseteq \aa$, and $\W$, a set of multi-dimensional range queries over $\bb$.  We call $\W$ the {\em workload}.  For example, suppose the database $\db$ contains records from the US Census describing various demographic characteristics of US citizens.  The analyst might specify $\bb = \set{age, salary}$ and a set $\W$ where each query is of the form, 
\begin{align*}
& \text{\sf select count(*) from R} \\
& \text{\sf where $a_{low} \leq$ age $\leq a_{high}$ and $s_{low} \leq$ salary $\leq s_{high}$} 	
\end{align*}
with different values for $a_{low}$, $a_{high}$, $s_{low}$, $s_{high}$.   We restrict our attention to the setting where the dimensionality, $k = |\bb|$, is small (our experiments report on $k \in \set{1, 2}$). 
All the differentially private algorithms considered in this paper attempt to answer the range queries in $\W$ on the private database $\db$ while incurring as little error as possible.

{ 

In this paper, we will often represent the database as a multi-dimensional array $\mathbf{x}$ of counts. 
For $\bb = \set{ B_1, \dots, B_k }$, let $n_j$ denote the domain size of $B_j$ for $j \in [1,k]$.  Then $\mathbf{x}$ has $(n_1 \times n_2 \times \ldots \times n_k)$ cells and the count in the $(i_1, i_2, \ldots, i_k)^{th}$ cell is
}
\begin{align*}
&\text{\sf select count(*) from R} \\
&\text{\sf where $B_1 = i_1$ and $B_2 = i_2$ and \ldots $B_k = i_k$}
\end{align*}
To compute the answer to a query in $\W$, one can simply sum the corresponding entries in $\x$.  (Because they are range queries, the corresponding entries form a (hyper-)rectangle in $\x$.)

\begin{example}
Suppose $\bb$ has the attributes age and salary (in tens of thousands) with domains $[1,100]$ and $[1,50]$ respectively. Then $\mathbf{x}$ is a $100\times 50$ matrix. The $(25,10)^{th}$ entry is the number of tuples with age $25$ and salary \$100,000. 
\end{example}

We identify three key properties of $\x$, each of which significantly impacts the behavior of privacy algorithms.  The first is the {\em domain size}, $n$, which is equivalently the number of cells in $\x$ (i.e., $n = n_1 \times \dots \times n_k$).  The second is the {\em scale} of the dataset, which is the total number of tuples, or the sum of the counts in $\x$, which we write as $\scale{\x}$.  Finally, the {\em shape} of a dataset is denoted as $\p$ where $\p = \x/\Lone{\x} = [p_1, \dots, p_n]$ is a non-negative vector that sums to 1.  The shape captures how the data is distributed over the domain and is independent of scale.

%!TEX root = paper.tex

\section{Algorithms \& Prior Results}  \label{sec:algorithm_background}

\subsection{Overview of Algorithm Strategies}
\label{sec:algorithm_overview}

%!TEX root = paper.tex

\begin{table}	
	\centering
	{\scriptsize
	\begin{tabular}{|l|c|c|l|l|l|c|c|}
	\hline
	 &  \multicolumn{5}{|c|}{\em Properties} & \multicolumn{2}{|c|}{\em Analysis} \\ \hline
	                   & H  &  P & \text{Dimen-} & \text{Param-}  & \text{Side} & \text{Consis-} & \text{Scale-$\epsilon$} \\ 
	\text{Algorithm}   &    &    & \text{sion}   &  \text{eters}              & \text{info} & \text{tent} & \text{Exch.} \\ \hline\hline
	\multicolumn{8}{|l|}{\em Data-independent} \\ \hline
	\hspace{0.05em} \identity~\cite{dwork2006calibrating}                    &   &     & Multi-D & -- & & yes & yes \\ \hline
	\hspace{0.05em} \privelet~\cite{xiao2010differential}                    & X &     & Multi-D & -- & & yes & yes \\ \hline
	\hspace{0.05em} \Hier~\cite{hay2010boosting}                             & X &     & 1D & $b=2$  & & yes & yes \\ \hline
	\hspace{0.05em} \Hierb~\cite{Qardaji13Understanding}                     & X &     & Multi-D & -- & & yes & yes \\ \hline
	\hspace{0.05em} \greedyH~\cite{Li14Data-}                              & X &     & 1D, 2D & $b=2$  & & yes & yes \\ \hline
	\multicolumn{8}{|l|}{\em Data-dependent} \\ \hline  
	\hspace{0.05em} \uniform                                                 &   & $\sim$  & Multi-D & -- & & no & yes \\ \hline
	\hspace{0.05em} \mwem~\cite{hardt2012a-simple}                           &   &     & Multi-D & $T$ & scale & no & yes \\ \hline 
	\hspace{0.05em} \mwemv                                                   &   &     & Multi-D & -- &  & no & yes \\ \hline 
	\hspace{0.05em} \ahp~\cite{zhangtowards}                                 &   & X   & Multi-D & $\propPart, \eta$ & & yes & yes \\ \hline 
	\hspace{0.05em} \ahpv                                                    &   & X   & Multi-D & -- & & yes & yes \\ \hline 
	\multirow{2}{*}{\hspace{0.05em} \dpcube~\cite{xiao2014dpcube}}           & \multirow{2}{*}{$\sim$} & \multirow{2}{*}{X}   & \multirow{2}{*}{Multi-D} & $\propPart=.5$,  & & \multirow{2}{*}{yes} & \multirow{2}{*}{yes} \\ 
	                            &  &  & & $n_p = 10$  & &  &  \\ \hline
	\multirow{2}{*}{\hspace{0.05em} \dawa~\cite{Li14Data-}}	& \multirow{2}{*}{X} & \multirow{2}{*}{X}   & \multirow{2}{*}{1D, 2D} & $\propPart=.25$, & & \multirow{2}{*}{yes} & \multirow{2}{*}{yes} \\ 
	 & & & & $b=2$ & & & \\ \hline
	\hspace{0.05em} \quadtree~\cite{Cormode11Differentially}                 & X & X   & 2D & c=10  & & no$^*$ & yes \\ \hline
		\hspace{0.05em} \UG~\cite{qardaji2013differentially}                     &   & X   & 2D & $c=10$  & scale & yes & yes \\ \hline
	\multirow{3}{*}{\hspace{0.05em} \AG~\cite{qardaji2013differentially}}                     & \multirow{3}{*}{$\sim$}     & \multirow{3}{*}{X}   & \multirow{3}{*}{2D} & $c = 10$,  & \multirow{3}{*}{scale} & \multirow{3}{*}{yes} & \multirow{3}{*}{yes} \\ 
	                     &      &    &  & $c_2 = 5$,  &  &  & \\ 
	                     &      &    &  & $\propPart=.5$  &  &  & \\ \hline
	\hspace{0.05em} \php~\cite{Acs2012compression}                           &   & X   & 1D & $\propPart=.5$  & & no & yes \\ \hline 
	\hspace{0.05em} \efpa~\cite{Acs2012compression}                          &   &     & 1D & -- & & yes & yes \\ \hline
	\hspace{0.05em} \structurefirst~\cite{xu2013differential}                &   & X   & 1D & $\propPart, k, F$ & scale & yes$^*$ & no \\ \hline 
	\end{tabular}
	}
	\caption{\label{tbl:algorithms} Algorithms evaluated in benchmark. Property column H indicates hierarchical algorithms and P indicates partitioning.  Parameters without assignments are ones that remain free.  Side information is discussed in \cref{sec:end-to-end}.  Analysis columns are discussed in \cref{sec:scale-eps} and \cref{sec:sub:var}.  Algorithm variants \mwemv and \ahpv are explained in \cref{sec:violators}.
	}
\end{table}

The algorithms evaluated in this paper are listed in \cref{tbl:algorithms}.  For each algorithm, the table identifies the dataset dimensionality it supports as well as other key properties (discussed further below).  In addition, it identifies algorithm-specific parameters as well as the possible use of ``side information'' (discussed in \cref{sec:principles}).  The table also summarizes our theoretical analysis, which is described in detail later (\cref{sec:scale-eps,sec:sub:var}).  Descriptions of individual algorithms are provided in \cref{sec:alg-desc}.

In this section, we categorize algorithms as either data-independent or data-dependent, and further highlight some key strategies employed, such as the use of hierarchical aggregations and partitioning.  In addition, we also illustrate how algorithm behavior is affected by properties of the input including dataset shape, scale, and domain size.

First, we describe a simple baseline strategy: release $\x$ after adding independent random noise to each count in $\x$.  To ensure differential privacy, the noise distribution is calibrated according to \cref{def:laplace_mechanism}, and each cell receives independent Laplace noise with a scale of $1/\epsilon$.  The limitation with this simple strategy is that when answering range queries, the variance in the answer increases linearly with the number of cells that fall within the range.  For large ranges, the error becomes intolerably high.  Thus, the performance of this approach depends critically on the {\em domain size}.

{\bf Hierarchical aggregation:} 
To mitigate the noise accumulation, several approaches not only obtain noisy counts for the individual cells in $\x$, but also obtain noisy estimates for the total count in hierarchically grouped subsets of cells.  Because each record is now being counted multiple times, the noise must be proportionally increased.  However, it has been shown in the 1D case that the increase is only logarithmic in the domain size and, further, any range query now requires summing only a logarithmic, rather than linear, number of noisy counts~\cite{hay2010boosting,xiao2011ireduct}.  For higher dimensions, the relative benefit of hierarchical aggregations diminishes~\cite{Qardaji13Understanding}.  
Algorithms employing hierarchical aggregations include \privelet~\cite{xiao2010differential}, \Hier~\cite{hay2010boosting}, \Hierb~\cite{Qardaji13Understanding}, \greedyH~\cite{Li14Data-}, \dawa~\cite{Li14Data-}, \AG~\cite{qardaji2013differentially}, \dpcube~\cite{xiao2014dpcube}, \quadtree~\cite{Cormode11Differentially}, \hybridtree~\cite{Cormode11Differentially}.

For many of the algorithms mentioned above, the noise added does not depend on the data. Thus, the performance is the same on all datasets of a given domain size.  Following~\cite{Li14Data-}, an algorithm whose error rate is the same for all possible datasets on a given domain is characterized as {\em data independent}.  \cref{tbl:algorithms} indicates which algorithms are data-independent.

It can be shown that all of the data independent algorithms studied here are instances of the matrix mechanism~\cite{Li:2010Optimizing-Linear,li2015matrix}, a generic framework that computes linear combinations of cell counts, adds noise, and then reconstructs estimates for the individual cells through linear transformations.  The algorithms differ in the choice of linear combination and thus experience different error rates on a given workload $\W$.  (Computing the optimal linear combination is computationally infeasible.) 

Thus, for data-independent algorithms, the only property of the input that affects performance is domain size.  We next describe strategies of data-dependent algorithms, whose performance is affected by dataset shape and scale.  

{\bf Partitioning:} One kind of data-dependent algorithm is one that partitions the data.  These algorithms  
reduce the error due to noise by only computing noisy aggregations of cell groups.  An example of such an approach is an equi-width histogram: the domain is partitioned into disjoint intervals (called ``buckets'') of equal size and a noisy count for each bucket is obtained.  To approximate the count for a cell within a bucket, the standard {\em assumption of uniformity} is employed.  The success of such a strategy depends critically on the {\em shape} of the dataset: If the distribution of cell counts within a bucket are nearly uniform, the noise is effectively reduced because each cell only receives a fraction of the noise; on the other hand, a non-uniform distribution implies high error due to approximating each cell by the bucket average.  Algorithms that partition include \AG, \hybridtree, \quadtree, \dpcube, \dawa, as well as \php~\cite{Acs2012compression}, \ahp~\cite{zhangtowards}, \structurefirst~\cite{xu2013differential}, \UG~\cite{qardaji2013differentially}.

Equi-width partitions are used by \quadtree, \UG, and \AG whereas other algorithms select the partition adaptively based on the characteristics of the data.  This is non-trivial because one must prove that the adaptive component of the algorithm satisfies differential privacy (and thus does not adapt too specifically to the input).  
Other algorithms, such as \mwem and \efpa~\cite{Acs2012compression}, adapt to the data as well, but use a different strategy than partitioning.

While one might expect that the performance of data-dependent algorithms is affected by dataset shape, a novel contribution of this paper is to show that it is also affected by {\em scale}.  The intuition for this can be seen by considering equi-width histograms.  Holding shape fixed, as scale increases, any deviations from a uniform distribution become magnified.  Thus, the relative benefit of partitioning cells into buckets diminishes.  While data adaptive approaches could, in principle, adjust with increasing scale, our findings and theoretical analysis show that this is not always the case.

{\bf Baselines:} \identity is a data-independent algorithm described earlier: it  adds noise to each cell count in $\x$, and is equivalent to applying the Laplace mechanism on the function that transforms $\db$ into $\x$.
\uniform is the second baseline, which uses its privacy budget to estimate the number of tuples (scale) and then produces a data-dependent estimate of the dataset by assuming uniformity.  It is equivalent to an equi-width histogram that contains a single bucket as wide as the entire domain.

With each baseline, the workload queries can be answered by summing the corresponding noisy cell counts.

\subsection{Empirical Comparisons in Prior Work}  \label{sec:prior_results}

We review prior studies with a focus on two things.  First, we characterize what is known about the state of the art for both 1- and 2D settings and identify gaps and inconsistencies.  Second, we look at the algorithmic strategies described in \cref{sec:algorithm_overview} and what is known about their effect on performance.

For the 1D setting, the study of Qardaji et al.~\cite{Qardaji13Understanding} suggests that \Hierb generally is the best data-independent algorithm, achieving lower error than many competing techniques, including \identity, \Hier, \privelet, \structurefirst, and others~\cite{li2015matrix,Cormode11Differentially,chaopvldb12,Yuan12Low-Rank}. However, their study does not include any data-dependent algorithms other than \structurefirst.  Thus, their study does not address whether \Hierb can outperform the data-dependent algorithms listed in \cref{tbl:algorithms}, a gap in knowledge that has yet to be resolved.

The results of Hardt et al.~\cite{hardt2012a-simple} suggest that \Hierb and the other data-independent techniques studied in \cite{Qardaji13Understanding} would be outperformed by \mwem for both 1- and 2D range queries.  In their study, \mwem almost always achieves an error rate that is less than a lower bound that applies to any instance of the matrix mechanism~\cite{li2015matrix,Li:2010Optimizing-Linear}, a generic class of differentially private algorithms that includes the top performers from the study of Qardaji et al.~\cite{Qardaji13Understanding} and all of the data-independent techniques listed in \cref{tbl:algorithms}.

Data-dependent and data-independent techniques are also compared in Li et al.~\cite{Li14Data-} and their findings are inconsistent with the results of Hardt et al.  Specifically, \mwem and other data-dependent techniques do not always outperform matrix mechanism techniques; for some ``hard'' datasets the error of \mwem is 2-$10\times$ higher than \privelet, one of the matrix mechanism instances included in their study.  The authors investigate the difficulty of datasets, but do not identify what properties of the input affect the performance of data-dependent algorithms.  In this paper, by explicitly controlling for scale and shape, we are able to resolve this gap and offer a much richer characterization of when data-dependent algorithms do well and explain the apparent contradiction between \cite{hardt2012a-simple} and \cite{Li14Data-}. 

For the 2D setting, we are not aware of any work that offers a comprehensive evaluation of all available techniques.  The 2D study of Qardaji et al.~\cite{Qardaji13Understanding} only compares \Hierb against \identity.  A second work by Qardaji et al.~\cite{qardaji2013differentially} compare their proposed algorithms with four techniques from the literature, and so is not as comprehensive as our study here.  Li et al.~\cite{Li14Data-} which appeared later did not compare to~\cite{qardaji2013differentially}.  Thus, another gap in the literature is a clear picture of the state of the art for 2D.

Finally, we highlight key findings regarding the algorithmic strategies discussed in \cref{sec:algorithm_overview}.  For hierarchical aggregation, Qardaji et al.~\cite{Qardaji13Understanding} carefully consider the effect of domain size and show that hierarchies can achieve significantly lower error than ``flat'' approaches like \identity but only when the domain size is sufficiently large.  Further, the minimum domain size increases exponentially with dimensionality.  
It must be at least 45 for 1D; at least $64^2$ for 2D, and at least $121^3$ for 3D.  For the data-dependent strategy of partitioning, a number of works show that it leads to low error for some datasets~\cite{Acs2012compression,Li14Data-,zhangtowards,xu2013differential,qardaji2013differentially,Cormode11Differentially,xiao2014dpcube}; however, we are not aware of any work that characterizes precisely when it improves upon data-independent techniques.

In summary, the existing literature has shown that data-dependent algorithms can sometimes do well and that domain size is an important factor in algorithm performance.  However, it has not clarified whether the best data-dependent algorithms can consistently outperform data-independent algorithms, and also leaves open the question of what factors, beyond domain size, affect  performance.

%!TEX root = paper.tex

\section{Evaluation Principles}\label{sec:principles}
In this section we present  fundamental principles that our evaluation framework \dpbench (see \cref{sec:benchmark}) is based on. These are principles that are often not followed in prior work, and when unheeded, can result in (a) gaps/incomplete understanding of the performance of differentially private algorithms, or (b) unfair comparison of algorithms. Our principles are categorized into three classes: {\em diversity of inputs}, {\em end-to-end private algorithms} and {\em sound  evaluation of outputs}. 

\subsection{Diversity of Inputs}
The first set of principles pertain to key inputs to the algorithm: the private data vector $\x$ and the differential privacy parameter $\epsilon$.  The intent of the principles is to ensure that algorithms are evaluated across a diverse set of inputs in order to provide a data owner with a comprehensive picture of algorithm performance.
Almost all work adheres to \cref{pr:eps}, but we state it for completeness.
\begin{principle}[$\epsilon$ Diversity]\label{pr:eps}
Evaluate algorithms under different $\epsilon$ values.
\end{principle}

In addition, most empirical work compares algorithms on several datasets.  We identify three critical characteristics of the dataset -- {\em domain size}, {\em scale}, and {\em shape}, as described in \cref{sec:data_model} -- each of which significantly impacts the behavior of privacy algorithms.  A novelty of our evaluation approach is to study the sensitivity of algorithm error to each of these properties (while holding the other two constant).

\begin{principle}[Scale Diversity] \label{pr:scale-div}
Evaluate algorithms using datasets of varying scales.
\end{principle}

\begin{principle}[Shape Diversity] \label{pr:shape-div}
Evaluate algorithms using datasets of varying shapes.
\end{principle}

\begin{principle}[Domain Size Diversity]\label{pr:domain-div}
Evaluate algorithms using datasets (and workloads) defined over varying domain sizes.
\end{principle}

We cannot anticipate in advance what dataset the algorithm might be run on, and what settings of $\epsilon$ a data owner may select. A sound evaluation should be able to predict algorithm performance on as many values of $\epsilon$, scale, shape and domain size as possible. This is especially important for data-dependent algorithms as their performance changes not only quantitatively but also {\em qualitatively} when changing these settings.

Our empirical comparisons of differentially private algorithms for range queries (\cref{sec:evaluation}), in both the one- and two-dimensional cases, show that the algorithms offering least error vary depending on the values of $\epsilon$, scale, shape, and domain size.  Because current algorithms offer no clear winner, evaluation on diverse inputs remains very important.

\stitle{Remarks on Prior Work} 
Empirical evaluations in prior work only compared algorithms on a single shape (e.g., \cite{Cormode11Differentially,xiao2014dpcube}), or a single domain size (e.g., \cite{Li14Data-}), or on multiple datasets without controlling for each of the input characteristics (e.g. \cite{Li14Data-, Acs2012compression, xu2013differential, zhangtowards,hardt2012a-simple}). These result in some of the gaps/inconsistencies in algorithm performance discussed in \cref{sec:prior_results}. 

\subsection{End-to-End Private Algorithms}  \label{sec:end-to-end}

The next three principles require that algorithms provide equivalent privacy guarantees to ensure fair comparisons.  While all published algorithms include a proof of differential privacy, some actions taken during deployment or evaluation of the algorithm can undermine that guarantee in subtle ways.

\begin{principle}[Private pre- and post-processing]  \label{pr:preprocess}
Any computation on the input dataset must be accounted for in the overall privacy guarantee of the algorithm being evaluated.
\end{principle}

Any computation that comes before or after the execution of the private algorithm must also be considered in the analysis of privacy.  This includes pre-processing steps such as data cleaning as well as any post-processing (e.g., selecting the ``best'' output from among a set of outputs).  If necessary, a portion of the user's $\epsilon$ privacy budget can be devoted to these tasks.

\begin{principle}[No Free Parameters]  \label{pr:param}
Include in every algorithm definition a data-independent or differentially private method for setting each required parameter.
\end{principle}

An important special case of preprocessing is parameter selection. Many published algorithms have additional parameters (beyond $\epsilon$) whose setting can have a large impact on the algorithm's behavior.  Borrowing a notion from physics, a parameter is considered {\em free} if it can be adjusted to fit the data -- in other words, the value of the parameter that optimizes algorithm performance is data-dependent.
Free parameters are problematic if they are later set by tuning them to the input data.  In a production environment, it risks violating the differential privacy guarantee.  In a research environment, it may impede a fair evaluation of algorithm performance, as some algorithms may be more sensitive to tuning than others. A principled methodology of parameter tuning is proposed in \cref{sec:freeparam-repair} and evaluated in \cref{sec:findings_tuning}.

\begin{principle}[Knowledge of Side Information] \label{pr:sideinfo} 
Use \\public knowledge of side information about the input dataset (e.g., scale) consistently across algorithms.
\end{principle}
Some algorithms assume that certain properties of the input dataset (e.g., scale) are considered public and use this during the execution of the algorithm.  Other algorithms do not use this information, or allocate some privacy budget to estimate it, leading to inconsistent comparisons.  In general, if the availability of ``side information'' is not accounted for properly, it can weaken the privacy guarantee.

\stitle{Remarks on Prior Work} 
Some work on differentially private machine learning pre-processes the input dataset using non-private data transformation algorithms~\cite{chaudhuri2013a-stability-based,fredrikson14usenix}.  In other work, parameters were tuned on the same datasets on which the algorithm's performance is later evaluated~\cite{hardt2012a-simple, zhangtowards,xu2013differential,Cormode11Differentially,xiao2011ireduct}.  Moreover, a few algorithms~\cite{hardt2012a-simple,xu2013differential,qardaji2013differentially} assume that the scale of the dataset is known.

\subsection{Sound Evaluation of Outputs}

The last four principles pertain to algorithm evaluation.

\begin{principle}[Measurement of variability]\label{pr:variation}
Algorithm output should be evaluated by expected error as well as a measure of error variability.
\end{principle}

The error of a differentially private algorithm is a random variable.  Most evaluations report its expectation or estimate it empirically.  However, when the algorithm is deployed, a user will apply it to private data and receive a single output.  Thus, it is also important to measure variability of error around its expectation. Note that this is different from reporting error bars representing standard error, which measures the uncertainty in the estimate of mean error and shrinks with increasing trials.  Instead, we want to measure how much the error of any single output might deviate from the mean.

\begin{principle}[Measurement of Bias]\label{pr:bias}
Algorithm output should be evaluated for bias in the answer.
\end{principle}

In addition to the expected error, a sound evaluation should also check whether the algorithm's output is biased. For instance, the Laplace mechanism is unbiased; i.e., the expected value of the output of the Laplace mechanism is equal to the original input. An attractive property of unbiased algorithms like the Laplace mechanism is that they are consistent -- as $\epsilon$ tends to $\infty$, the error tends to $0$.
Recent data-dependent algorithms incur less error than the Laplace mechanism by introducing some bias in the output. Understanding the bias is important; we show in this paper (\cref{sec:sub:var}) that some algorithms incur a bias (and therefore incur error) even when $\epsilon$ tends to infinity.

\begin{principle}[Reasonable privacy and utility] \label{pr:reasonable}
{\color{white}.} \\Evaluate algorithms for input settings that result in reasonable privacy and utility.
\end{principle}

Under extreme settings of the input parameters, one algorithm may outperform another, but both algorithms may produce useless results.  Similarly, it is not meaningful to make relative comparisons of algorithm performance when a reasonable privacy guarantee is not offered.

\stitle{Remarks on Prior Work} 
Almost all prior empirical work on differential privacy (with the exception of those that prove theoretical bounds on the error in terms of $(\alpha, \delta)$-usefulness~\cite{blum2013learning}) focus on mean error and not its variability. Algorithm bias is seldom analyzed.  While most algorithms are explicitly compared to a baseline like \identity (\cite{Cormode11Differentially} is an exception), none of algorithms have been compared with a data-dependent baseline  like \uniform.

%!TEX root = paper.tex
\section{DPBench Evaluation Framework}  \label{sec:benchmark}

In this section, we describe \dpbench, a framework for evaluating the accuracy of differentially private algorithms. \dpbench's goal is to formulate a set of standards for empirical evaluation that adhere to the principles from \cref{sec:principles}.
The framework can be applied to a number of analysis tasks.  Because  
different analysis tasks require the use of disparate algorithms, it is necessary to instantiate distinct benchmarks for each kind of task.\footnote{This is akin to the various TPC benchmarks for evaluating database performance.} \dpbench aims to provide a common framework for developing these benchmarks.

We define a benchmark as a 9-tuple $\{{\cal T}, {\cal W}, {\cal D}, {\cal M}, {\cal L}, {\cal G}, {\cal R}, {\cal E}_M, {\cal E}_I\}$. ${\cal T}$ denotes the task that the benchmark is associated with. ${\cal W}$ denotes the set of query workloads that are representative for the task ${\cal T}$. For instance, if ${\cal T}$ is the task of ``answering one-dimensional range queries'', ${\cal W}$ could contain a workload of random range queries. Note that one could use more than one workload to evaluate algorithms for a task. Evaluating differentially private algorithms is typically done on publicly available datasets, which we refer to as a {\em source} data instance $I$. ${\cal D}$ denotes a set of such source datasets.  ${\cal M}$ lists the set of algorithms that are compared by the benchmark. ${\cal L}$ denotes {\em loss functions} that are used to measure the error (or distance) between $\y = \W \x$, the true answer to a workload,  and $\esty$, the output of a differentially private algorithm in ${\cal M}$. An example of a loss function that we will consistently use in this paper is the $L_2$ norm of the difference between $\y$ and $\esty$; i.e. error is $\Ltwo{\W\x - \esty}$. 

We call ${\cal W}, {\cal D}, {\cal M}$ and ${\cal L}$ {\em task specific components} of a \dpbench benchmark, since they vary depending on the task under consideration. We discuss them in detail in \cref{sec:setup}. On the other hand the remaining components ${\cal G}, {\cal R}, {\cal E}_M, {\cal E}_I$ are task independent, and help benchmarks to adhere to the principles in \cref{sec:principles}. ${\cal G}$ is a {\em data generator} that ensures diversity of inputs (\crefrange{pr:eps}{pr:domain-div}). ${\cal R}$ defines a set of {\em algorithm repair functions} that ensure that the algorithms being evaluated satisfy end-to-end privacy (\crefrange{pr:preprocess}{pr:sideinfo}). In particular, \dpbench provides a technique to automatically train free parameters of an algorithm. ${\cal E}_M$ and ${\cal E}_I$ denote standards for measuring and interpreting empirical error. These help adhere to \crefrange{pr:variation}{pr:reasonable} by explicitly defining standards for measuring variation and bias, declaring winners among algorithms, and defining baselines for interpreting error. Next, we describe these in detail, and on the way define two novel theoretical properties {\em scale-epsilon exchangeability} and {\em consistency}.

\subsection{Data Generator {\cal G}}\label{sec:datagen}

Let $I \in {\cal D}$ be a source dataset with true domain $D_I = (A_1 \times \ldots A_k)$ and true scale $s_I$. The data generator ${\cal G}$ takes as input a scale $m$, and a new domain $D$, which may be constructed by (i) considering a subset of attributes $\bb$, and (ii) coarsening the domain of attributes in $\bb$. Note that $m$ and $D$ are possibly different from the true scale $s_I$ and domain $D_I$. ${\cal G}$ outputs a new data vector $\x$ with the new scale $m$ and new domain $D$ as follows. First, ${\cal G}$ computes a histogram $\x'$ of counts on the new domain based on $I$. Then, ${\cal G}$ normalizes $\x'$ by its true scale to isolate the {\em shape} $\p$ of the source data on the new domain $D$.  ${\cal G}$ generates $\x$ by sampling $m$ times (with replacement) from $\p$. If we set the desired scale to the true scale $s_I$ we will generate a data vector $\x$ that is approximately the same as the original (modulo sampling differences).

The output vector $\x$ of ${\cal G}$ is used as input to the privacy algorithms.  Diversity of scale (\cref{pr:scale-div}) is achieved by varying $m$.  Diversity of domain size (\cref{pr:domain-div}) is achieved by varying $D$.  Shape is primarily determined by $I$ but also impacted by the domain definition, $D$, so diversity of shape (\cref{pr:shape-div}) requires a variety of source datasets and coordinated domain selection.	

The data generator ${\cal G}$ outputs datasets that help evaluate algorithms on each of the data characteristics shape, scale and domain size while keeping the other two properties constant. Without such a data generator, the effects of shape and scale on algorithm error can't be distinguished, thus prior work was unable to correctly identify features of datasets that impacted error. For a fixed shape, increasing scale offers a stronger ``signal'' about the underlying properties of the data. In addition, this sampling strategy always results in datasets with integral counts (simply multiplying the distribution by some scale factor may not).

\subsection{Algorithm Repair Functions ${\cal R}$}\label{sec:repair}

While automatically verifying whether an algorithm performs pre- or post-processing that violates differential privacy is out of the scope of this benchmark, we discuss two {\em repair functions} to help adhere to the free parameters and side information principles (\cref{pr:param,pr:sideinfo}, respectively).

\stitle{Learning Free Parameter Settings $R_{param}$}\label{sec:freeparam-repair}
We present a first cut solution to handling free parameters. Let $\alg_{\theta}$ denote a private algorithm $\alg$ instantiated with a vector of free parameters $\theta$. 
We use a separate set of datasets to tune these parameters; these datasets will {\em not} be used in the evaluation.
Given a set of training datasets ${\cal D}_{train}$, we apply data generator ${\cal G}$ and learn a function $R_{param}: (\epsilon, \scale{\x}, n) \rightarrow \theta$ that given $\epsilon$, scale, and the domain size, outputs parameters $\theta$ that result in the lowest error for the algorithm. Given this function, the benchmark extends the algorithm by  adaptively selecting parameter settings based on scale and epsilon. If the parameter setting depends on scale, a part of the privacy budget is spent estimating scale, and this introduces a new free parameter, namely the budget spent for estimating scale. The best setting for this parameter can also be learned in a similar manner. 

\stitle{Side Information $R_{side}$}\label{sec:sideinfo-repair}
Algorithms which use non-private side information can typically be corrected by devoting a portion of the privacy budget to learning the required side information, then using the noisy value in place of the side information.  
This process is difficult to automate but may be possible with program analysis in some cases. 
This has the side-effect of introducing a new parameter which determines the fraction of the privacy budget to devote to this component of the algorithm, which in turn can be set using our learning algorithm  (described above).

\subsection{Standards for Measuring Error ${\cal E}_M$}\label{sec:measuring}
\stitle{Error}
\dpbench uses {\em scaled average per-query error} to quantify an algorithm's error on a workload. 
\begin{definition}[Scaled Average Per-Query Error]\label{def:sample-error}
Let $\W$ be a workload of $q$ queries, $\x$ a data vector and $s = \scale{\x}$ its scale. Let $\esty = \alg(\x, \W, \epsilon)$ denote the noisy output of algorithm $\alg$. Given a loss function $L$, we define {\em scale average per-query error} as $\frac{1}{s\cdot q} L(\esty, \W\x)$.
\end{definition}

By reporting scaled error, we avoid considering a fixed absolute error rate to be equivalent on a small scale dataset and a large scale dataset. For example, for a given workload query, an absolute error of 100 on a dataset of scale 1000 has very different implications than an absolute error of 100 for a dataset with scale 100,000. In our scaled terms, these common absolute errors would be clearly distinguished as 0.1 and 0.001 scaled error. Accordingly, scaled error can be interpreted in terms of population percentages. Using scaled error also helps us define the scale-epsilon exchangeability property in \cref{sec:scale-eps}.

Considering per-query error allows us to compare the error on different workloads of potentially different sizes. For instance, when examining the effect of domain size $n$ on the accuracy of algorithms answering the identity workload, the number of queries $q$ equals $n$ and hence would vary as $n$ varies.

\stitle{Measuring Error}
The error measure (\cref{def:sample-error}) is a random variable.  We can estimate properties such as their mean and variance through repeated executions of the algorithm.  In addition to comparing algorithms using mean error, \dpbench also compares algorithms based on the 95 percentile of the error. This takes into account the variability in the error (adhering to \cref{pr:variation}) and might be an appropriate measure for a ``risk averse'' analyst who prefers an algorithm with reliable performance over an algorithm that has lower mean performance but is more volatile. Means and 95 percentile error values are computed on multiple independent repetitions of the algorithm over multiple samples $\x$ drawn from the data generator to ensure high confidence estimates.

\dpbench also identifies algorithms that are competitive for state-of-the-art performance for each setting of scale, shape and domain size. An algorithm is {\em competitive} if it either (i) achieves the lowest error, or (ii) the difference between its error and the lowest error is not statistically significant. Significance is assessed using a unpaired t-test with a Bonferroni corrected $\alpha = 0.05/(n_{algs}-1)$, for running $(n_{algs}-1)$ tests in parallel. $n_{algs}$ denotes the number of algorithms being compared. Competitive algorithms can be chosen both based on mean error (a ``risk neutral'' analyst) and 95 percentile error (a ``risk averse'' analyst).  

\subsection{Standards for Interpreting Error ${\cal E}_I$}\label{sec:interpreting}

When drawing conclusions from experimental results, \cref{pr:reasonable} should be respected.  One way to assess reasonable utility is by comparing with appropriate baselines. 

We use \identity and \uniform (described in \cref{sec:algorithm_overview}) as upper-bound baselines. Since \identity is a straightforward application of the Laplace mechanism, we expect a more sophisticated algorithm to provide a substantial benefit over the error achievable with \identity. Similarly, \uniform learns very little about $\x$, only its scale.  An algorithm that offers error rates comparable or worse than \uniform is unlikely to provide useful information in practical settings. Note that there might be a few settings where these baselines can't be beaten (e.g., when shape of $\x$ is indeed uniform). However, an algorithm should be able to beat these baselines in a majority of settings.

\subsection{Scale Epsilon Exchangeability} \label{sec:scale-eps}
We describe next a novel property of most algorithms due to which increasing $\epsilon$ and scale both have an equivalent effect on algorithm error.  We say an algorithm is scale-epsilon {\em exchangeable} if  increasing scale by a multiplicative factor $c$ has a precisely equivalent effect, on the scaled error, as increasing epsilon by a factor of $c$. Exchangeability has an intuitive meaning: to get better accuracy from an exchangeable algorithm, a user can either acquire more data or, equivalently, find a way to increase their privacy budget.  Moreover, for such algorithms, diversity of scale (\cref{pr:scale-div}) implies diversity of $\epsilon$ (\cref{pr:eps}).
\vspace{-2mm}
\begin{definition}[Scale-epsilon exchangeability]  \label{def:exchangeability}
Let $\p$ be a shape and $\W$ a workload. If $\x_1 = m_1\p$ and $\x_2 = m_2\p$, then algorithm $\alg$ is scale-epsilon exchangeable if $error(\alg(\x_1,\W,\epsilon_1))=error(\alg(\x_2,\W,\epsilon_2))$ whenever $\epsilon_1m_1=\epsilon_2m_2$.  
\end{definition}

We prove in \cref{sec:scale-eps-proofs} that all but one of the algorithms we evaluate in this paper are scale-epsilon exchangeable.  The exception is \structurefirst \cite{xu2013differential}, which is not exchangeable but empirically behaves so.  An important side-effect of this insight is that in empirical evaluations one can equivalently vary scale rather than $\epsilon$ if the algorithm is proven to satisfy scale-$\epsilon$ exchangeability.

\subsection{Consistency}

As $\epsilon$ increases, the privacy guarantee weakens.  A desirable property of a differential privacy algorithm is that its error decreases with increasing $\epsilon$ and ultimately disappears as $\epsilon \rightarrow \infty$.  We call this property consistency. Algorithms that are not consistent are inherently biased. 
\vspace{-2mm}
\begin{definition}[Consistency] \label{def:consistency}
An algorithm satisfies consistency if the error for answering any workload tends to zero as $\epsilon$ tends to infinity.
\end{definition}

%!TEX root = paper.tex

\section{Experimental Setup} \label{sec:setup}

In this section we use our benchmark framework to define concrete benchmarks for 1- and 2-D range queries.  For each benchmark, we precisely describe the task-specific components of the benchmark $({\cal D}, {\cal W}, {\cal M}, {\cal L})$, namely the datasets, workloads, algorithms and loss functions with the goal of providing a reproducible standard of evaluation.  Experimental results appear in \cref{sec:evaluation}.

\subsection{Datasets}

\cref{tbl:datasets} is an overview of the datasets we consider.  11 of the datasets have been used to evaluate private algorithms in prior work.  We have introduced \hl{14} new datasets to increase shape diversity. Datasets are described in \cref{sec:dataset-desc}.  The table reports the original scale of each dataset.  We use the data generator ${\cal G}$ described before to generate datasets with scales of $\{10^3, 10^4, 10^5, 10^6, 10^7, 10^8\}$.  

The maximum domain size is 4096 for 1D datasets and $256 \times 256$ for 2D datasets.  The table also reports the fraction of cells in $\x$ that have a count of zero at this domain size.  By grouping adjacent buckets, we derive versions of each dataset with smaller domain sizes.  For 1D, the domain sizes are $\set{256, 512, 1024, 2048}$; for 2D, they are $\set{ 32 \times 32, 64 \times 64, 128 \times 128, 256 \times 256 }$.

For each scale and domain size, we randomly sample 5 data vectors from our data generator and for each data vector, we run the algorithms 10 times.

\begin{table}	
		\centering
	{\scriptsize
	\begin{tabular}{| l | l | l | l |} \hline
	{Dataset name} & {Original} & \% {Zero}   & {Previous} \\ 
	             & {Scale}      &  {Counts}   & {works}                \\ \hline \hline
	\multicolumn{4}{|l|}{\em 1D datasets} \\ \hline
	\adult      & 32,558    & 97.80\%  & \cite{hardt2012a-simple,Li14Data-} \\
	\hepph      & 347,414    & 21.17\% & \cite{Li14Data-} \\
	\income     & 20,787,122 & 44.97\% & \cite{Li14Data-} \\
	\medcost    & 9,415      & 74.80\% & \cite{Li14Data-} \\
	\nettrace   & 25,714     & 96.61\% & \cite{hay2010boosting,Acs2012compression,zhangtowards,xu2013differential} \\
	\patent     & 27,948,226 & 6.20\%  & \cite{Li14Data-} \\
	\searchlogs & 335,889    & 51.03\% & \cite{hay2010boosting,Acs2012compression,zhangtowards,xu2013differential} \\ 	\bidsfj &  1,901,799 &  0\% &     new \\
	\bidsfm & 2,126,344  &  0\% &     new \\
	\bidsall &  7,655,502 &  0\% &     new \\
	\mdsalary & 135,727 &  83.12\% &     new \\
	\mdsalaryfa & 100,534  &  83.17\% &     new \\
	\lcreqfa &  3,737,472 & 61.57\%  &     new \\
	\lcreqfb &  198,045 &  67.69\% &     new \\
	\lcreqall &  3,999,425 &  60.15\% &     new \\
	\lcdtirfa &  3,336,740 &  0\% &     new \\
	\lcdtirfb &  189,827 & 11.91\%  &     new \\
	\lcdtirall &  3,589,119 &  0\% &     new \\ \hline

	\multicolumn{4}{|l|}{\em 2D datasets} \\ \hline
	\TDbeijingtaxis &  4268780 & 78.17\% & \cite{he2015dpt} \\
	\TDbeijingtaxie &  4268780 & 76.83\% & \cite{he2015dpt} \\
	\TDcheckin & 6442863 & 88.92\% & \cite{qardaji2013differentially} \\ 
		\TDadult &  32561 & 99.30\% & \cite{hardt2012a-simple}  \\
	\TDcabspottings &  464040 & 95.04\% & \cite{epfl-mobility-20090224} \\
	\TDcabspottinge &  464040 & 97.31\% & \cite{epfl-mobility-20090224}  \\
	\TDmdsalary & 70526 & 97.89\% & new \\
	\TDloan &  550559 & 92.66\% & new \\
	\TDstoke &  19435 & 79.02\% & new \\

		\hline
	\end{tabular}
	}
	\caption{\label{tbl:datasets} Overview of datasets.} 
\end{table}	

\subsection{Workloads \& Loss Functions}

We evaluate our algorithms on different workloads of range queries.  For 1D, we primarily use the \prefixworkload workload, which consists of $n$ range queries $[1,i]$ for each $i \in [1, n]$. The \prefixworkload workload has the desirable property that any range query can be derived by combining the answers to exactly two queries from \prefixworkload. For 2D, we use 2000 random range queries as an approximation of the set of all range queries.

As mentioned in \cref{sec:benchmark}, we use $L_2$ as the loss function.

\subsection{Algorithms}

The algorithms compared are listed in \cref{tbl:algorithms}.  The dimension column indicates what dimensionalities the algorithm can support; algorithms labeled as Multi-D are included in both experiments.  Complete descriptions of algorithms appear in \cref{sec:alg-desc}.

\subsection{Resolving End-to-End Privacy Violations}  \label{sec:violators}

\noindent{\bf Inconsistent side information:} Recall that \cref{pr:sideinfo} prevents the inappropriate use of private side information by an algorithm.  \structurefirst, \mwem, \UG, and \AG assume the true scale of the dataset is known.  
{ To gauge any potential advantage gained from side information, we evaluated algorithm variants where a portion of the privacy budget, denoted $\propPart_{total}$, is used to noisily estimate the scale. 
To set $\propPart_{total}$, we evaluated the algorithms on synthetic data using varying values of $\propPart_{total}$.  In results not shown, we find that setting $\propPart_{total} = 0.05$ achieves reasonable performance.  For the most part, the effect is modestly increased error (presumably due to the reduced privacy budget available to the algorithm).  However, the error rate of \mwem increases significantly at small scales (suggesting it is benefiting from side information).  In \cref{sec:evaluation}, all results report performance of the original unmodified algorithms. While this gives a slight advantage to algorithms that use side information,
 it also faithfully represents the original algorithm design.
}

\vspace{2mm}
\noindent{\bf Illegal parameter setting}
\cref{tbl:algorithms} shows all the parameters used for each algorithm.  Parameters with assignments have been set according to fixed values provided by the authors of the algorithm.  Those without assignments are free parameters that were set in prior work in violation of \cref{pr:param}.  

For \mwem, the number of rounds $T$ is a free variable that has a major impact on \mwem's error.  According to a pre-print version of~\cite{hardt2012a-simple}, the best performing value of $T$ is used for each task considered.  For the one-dimensional range query task considered, $T$ is set to 10.  Similarly, for \ahp, two parameters are left free: $\eta$ and $\propPart$ which were tuned on the input data.   

To adhere to \cref{pr:param}, we use the learning algorithm for setting free parameters (\cref{sec:freeparam-repair}) to set free parameters for \mwem and \ahp. In our experiments, the extended versions of the algorithms are denoted \mwemv and \ahpv.  
{ In both cases, we train on shape distributions synthetically generated from power law and normal distributions. 

For \mwemv we determine experimentally the optimal
  $T \in [1, 200]$ for a range of $\epsilon$-scale products.  
}
 As a result, $T$ varies from 2 to 100 over the range of scales we consider.  This improves the performance of \mwem (versus a static setting of $T$) and does not violate our principles for private parameter setting.  The success of this method is an example of data-independent parameter setting.

\structurefirst requires three parameters: $\propPart$, $k$, $F$.  Parameter $F$ is free only in the sense that it is a function of scale, which is side information (as discussed above).  For $k$, the authors propose a recommendation of $k = \lceil\frac{n}{10}\rceil$ after evaluating various $k$ on input datasets.  Their evaluation, therefore, did not adhere to \cref{pr:param}.  However, because our evaluation uses different datasets, we can adopt their recommendation without violating \cref{pr:param} -- in effect, their experiment serves as a ``training phase'' for ours.  Finally, $\propPart$ is a function of $k$ and $F$, and thus is no longer free once those are fixed.

\subsection{Implementation Details}  \label{sec:impl_details}

We use implementations from the authors for \dawa, \greedyH, \Hier, \php, \efpa, and \structurefirst.  
We implemented \mwem, \Hierb, \privelet, \ahp, \dpcube, \AG, \UG and \quadtree ourselves in Python.  
All experiments are conducted on Linux machines running CentOS 2.6.32 (64-bit) with 16 Intel(R) Xeon(R) CPU E5-2643 0 @ 3.30GHz with 16G or 24G of RAM.

%!TEX root = paper.tex
\section{Experimental Findings} \label{sec:evaluation}

We present our findings for the 1D and 2D settings.  
For the 1D case, we evaluated \hl{14} algorithms on 18 different datasets, each at 6 different scales and 4 different domain sizes.  For 2D, we evaluated \hl{14} algorithms on 9 different datasets, each at 6 scales and 4 domain sizes.
In total we evaluated 
7,920 different experimental configurations. 
The total computational time was $\approx22$ days of single core computation.  Given space constraints, we present a subset of the experiments that illustrate  key findings.

In results not shown, we compare the error of data-independent algorithms (\Hierb, \greedyH, \Hier, \privelet, \identity) on various range query workloads.  The error of these algorithms is, of course, independent of the input data distribution, and the absolute error is independent of scale.  (Or, in scaled-error terms, it diminishes predictably with increasing scale for all algorithms).  The error for this class of techniques is  well-understood because it is more amenable to analysis than data-dependent techniques.  
In order to simplify the presentation of subsequent results, 
we {\em restrict our attention to two data-independent algorithms:} \identity, which is a useful baseline, and \Hierb, which is the best performing data-independent technique for 1D and 2D range query workloads with larger ranges.

\begin{figure*}[!t]
\centering
	\begin{subfigure}[b]{.49\textwidth}
		\includegraphics[width=\textwidth]{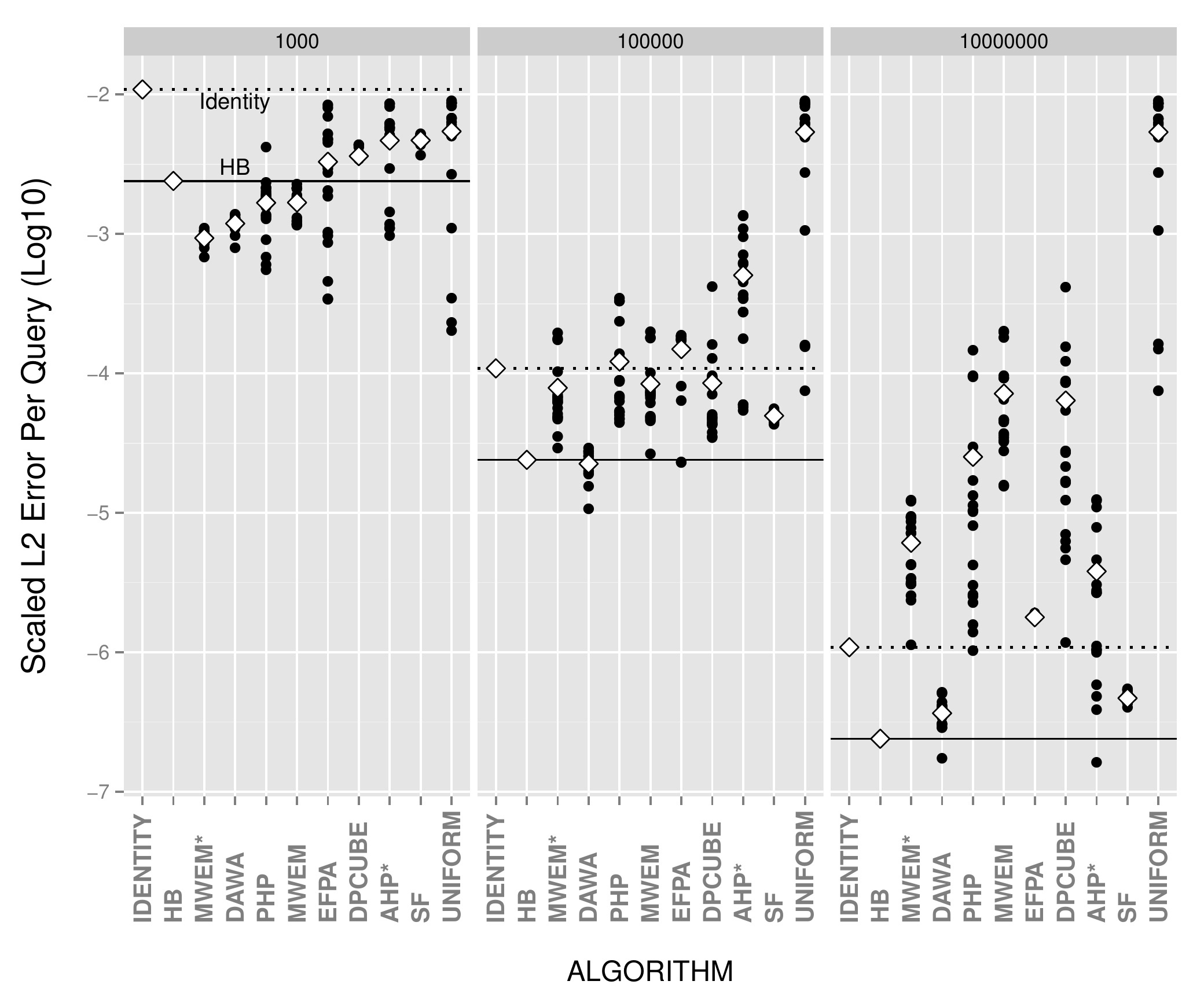}
		\caption{ \label{fig:shape-range-1D} 1D case: domain=4096, workload=Prefix}
	\end{subfigure}
	\begin{subfigure}[b]{.49\textwidth}
		\includegraphics[width=\textwidth]{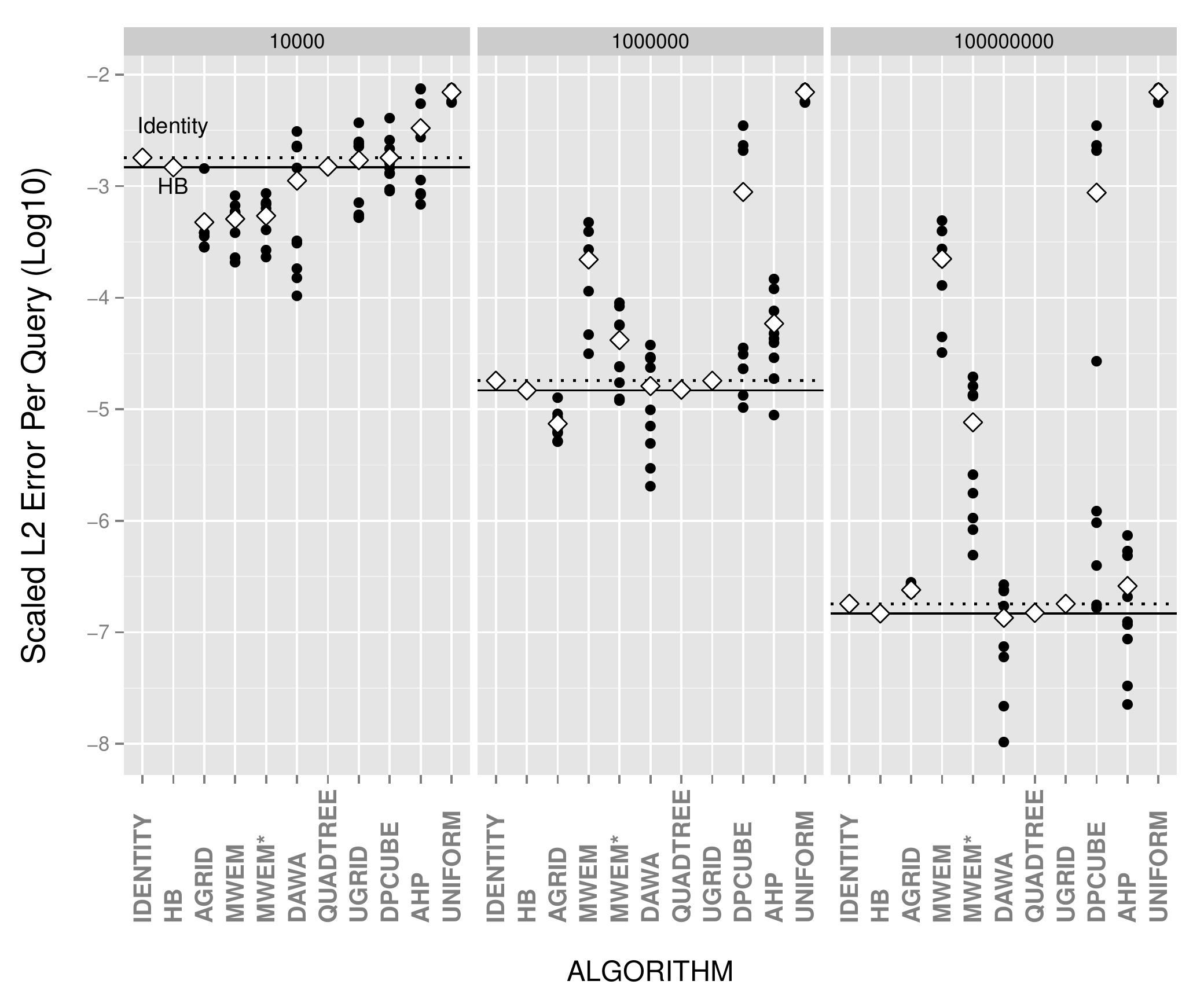}
		\caption{ \label{fig:shape-range-2D} 2D case: domain=128x128, workload=2000 random range queries}
	\end{subfigure}
	\caption{ \label{fig:shape-range} These figures show how average error rates of algorithms vary across the datasets described in \cref{tbl:datasets}, for both the 1D case (a) and 2D case (b).  Each black dot  ($\bullet$) corresponds to average error for one dataset; white diamonds ($\diamond$) correspond to the mean error over all datasets.  Horizontal lines are used to emphasize the error rates of the best and worst data-independent algorithms.}
\end{figure*}

\subsection{Impact of Input Properties on Performance}
	
As advocated in \crefrange{pr:eps}{pr:domain-div}, we now systematically examine the key properties of algorithm input that can affect performance: scale, domain size, shape, and $\epsilon$.  To aid in readability of figures, we report mean error without error bars but discuss variability over trials in detail in \cref{sec:sub:var}.

\noindent{\bf Scale:} \cref{fig:shape-range-1D,fig:shape-range-2D} compare algorithm performance at a fixed epsilon ($\epsilon=0.1$) and increasing scales. Because of the scale-epsilon exchangeability property (\cref{def:exchangeability}), which almost all algorithms satisfy,\footnote{\structurefirst is not exchangeable, but it empirically behaves so.} the figures look identical if we fix scale and show increasing $\epsilon$.  For example, the first panel of \cref{fig:shape-range-1D} reports error for scale=1000 and $\epsilon=.1$ but the plotted error rates and relationships are identical to those for scale=$10,000$ and $\epsilon=.01$, scale=100,000 and $\epsilon=.001$, etc.  \cref{fig:shape-range-1D,fig:shape-range-2D} also capture variation across datasets (each black dot is the error of an algorithm on a particular dataset at the given scale; the white diamond is the algorithm's error averaged over all datasets).

\begin{finding}[Benefits of Data-Dependence\\] \label{finding:data-dep}
Data-dep\-endence can offer significant improvements in error, especially at smaller scales or lower epsilon values.  
\end{finding}

The figures show that at small scales, almost all data-dependent algorithms outperform the best data-independent algorithm, \Hierb, on at least one dataset (black dot), sometimes by an order of magnitude.  Further, the best data-dependent algorithms beat the best data-independent algorithms on average across all datasets (white diamonds), often by a significant margin. (For 1D, at the smallest scale, it is as much as a factor of 2.47; for 2D, it is up to a factor of 3.10.)

\begin{finding}[Penalty of Data Dependence] 
Some data-dependent algorithms break down at larger scales (or higher $\epsilon$).
\end{finding}

Despite the good performance of data-dependent algorithms at small scales, at moderate or larger scales, many data-dependent algorithms have error significantly larger than \identity, the data-independent baseline, which indicates quite disappointing performance.  
Even for the best data-dependent methods, the comparative advantage of data-dependence decreases and, for the most part, eventually disappears as scale increases.  At the largest scales in both 1D and 2D, almost all data-dependent algorithms are beaten by \Hierb.  At these scales, the only algorithms to beat \Hierb are \dawa and \ahpv.  In the 1D case, this only happens on a couple of datasets; in the 2D case, it happens on about half the datasets.

\noindent{\bf Shape:} 
In \cref{fig:shape-range-1D,fig:shape-range-2D}, the effect of shape is evident in the spread of black dots for a given algorithm.  

\begin{finding}[Variation of Error with Data Shape] \label{finding:shape}
\qquad\: Algorithm error varies significantly with dataset shape, even for a fixed scale, and algorithms differ on the dataset shapes on which they perform well.
\end{finding}

Even when we control for scale, each data-dependent algorithm displays significant variation of error across datasets.  In the 1D case (\cref{fig:shape-range-1D}) at the smallest scale, \efpa's error varies by a factor of 24.88 from lowest error to highest.  Across all scales, algorithms like \mwem, \mwemv, \php, \efpa, and \ahpv show significantly different error rates across shape.  \structurefirst and \dawa offer the lowest variation across dataset shapes.  In the 2D case (\cref{fig:shape-range-2D}), we also see significant variation at all scales.  (Interestingly, \dawa exhibits high variation in 2D.)

While \cref{fig:shape-range-1D,fig:shape-range-2D} show variation, one cannot compare algorithms across shapes.  To do that, we use \cref{fig:shape1D,fig:shape2D}, where scale is fixed and shape is varied.  (For readability, we show only a subset of algorithms as explained in the caption for \cref{fig:shape-and-domain}.)  These figures show that the comparative performance of algorithms varies across shape.  \cref{fig:shape1D} shows that four different algorithms achieve the lowest error on some shape.  In addition, a dataset that is ``easy'' for one algorithm appears ``hard'' for another: witness \efpa is able to exploit the shape of \bidsall yet suffers on \nettrace; the opposite is true for \mwem.  
In \cref{fig:shape2D}, the shapes where \dawa performs poorly, \AG does quite well.
These results suggest that data-dependent algorithms are exploiting different properties of the dataset.  

The unpredictable performance across data shape is a challenge in deployment scenarios because privacy algorithms cannot be selected by looking at the input data and because data-dependent algorithms rarely come with public error bounds.  We discuss these issues further in \cref{sec:discussion}.

\begin{figure*}[t]
   \begin{subfigure}[b]{0.33\textwidth}    \centering
		\includegraphics[width=\textwidth]{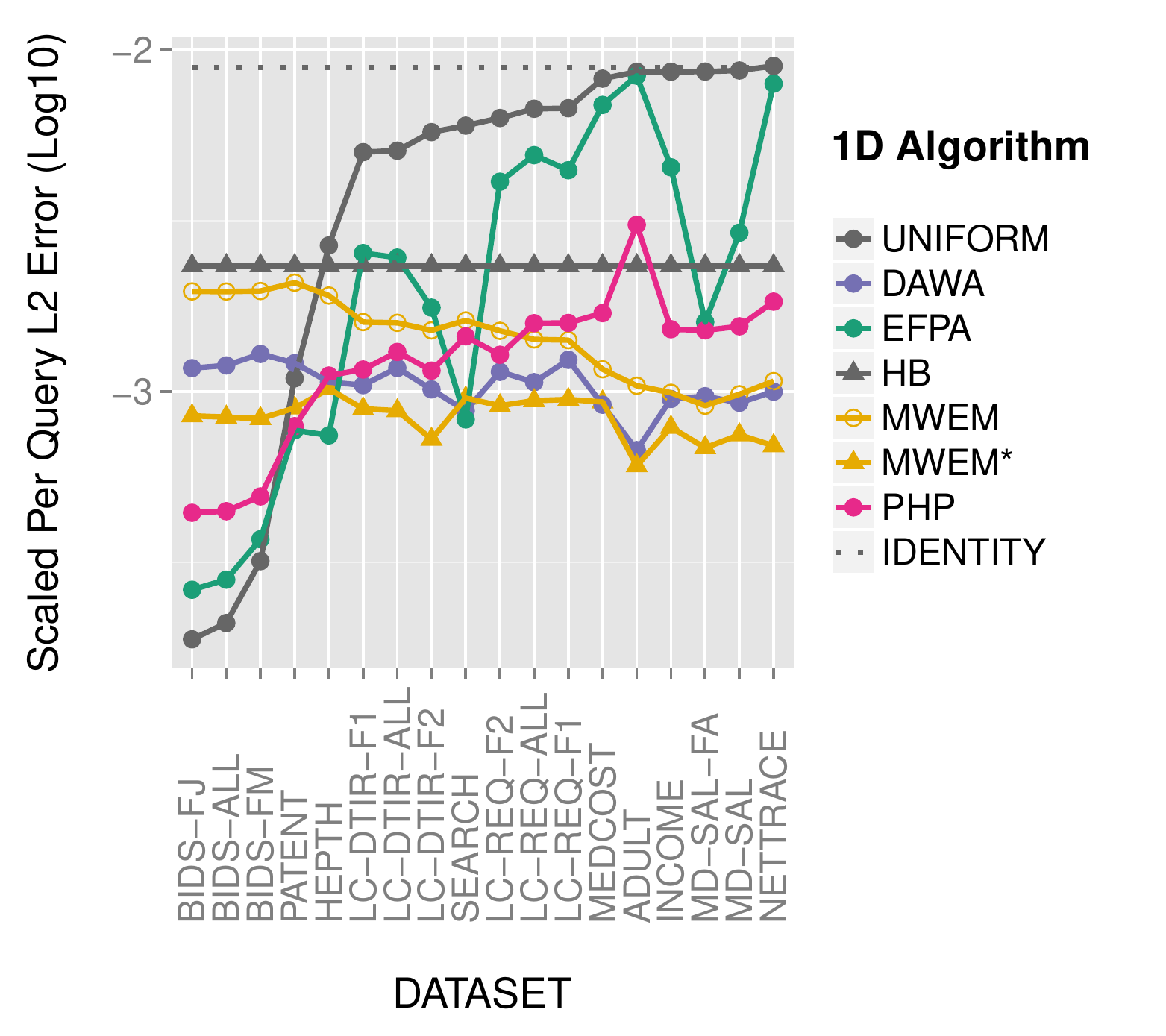}
		\caption{\label{fig:shape1D} Error by shape (1D, scale$=10^3$, domain size$=4096$)}
   \end{subfigure}
   \begin{subfigure}[b]{0.33\textwidth}    \centering		
		\includegraphics[width=\textwidth]{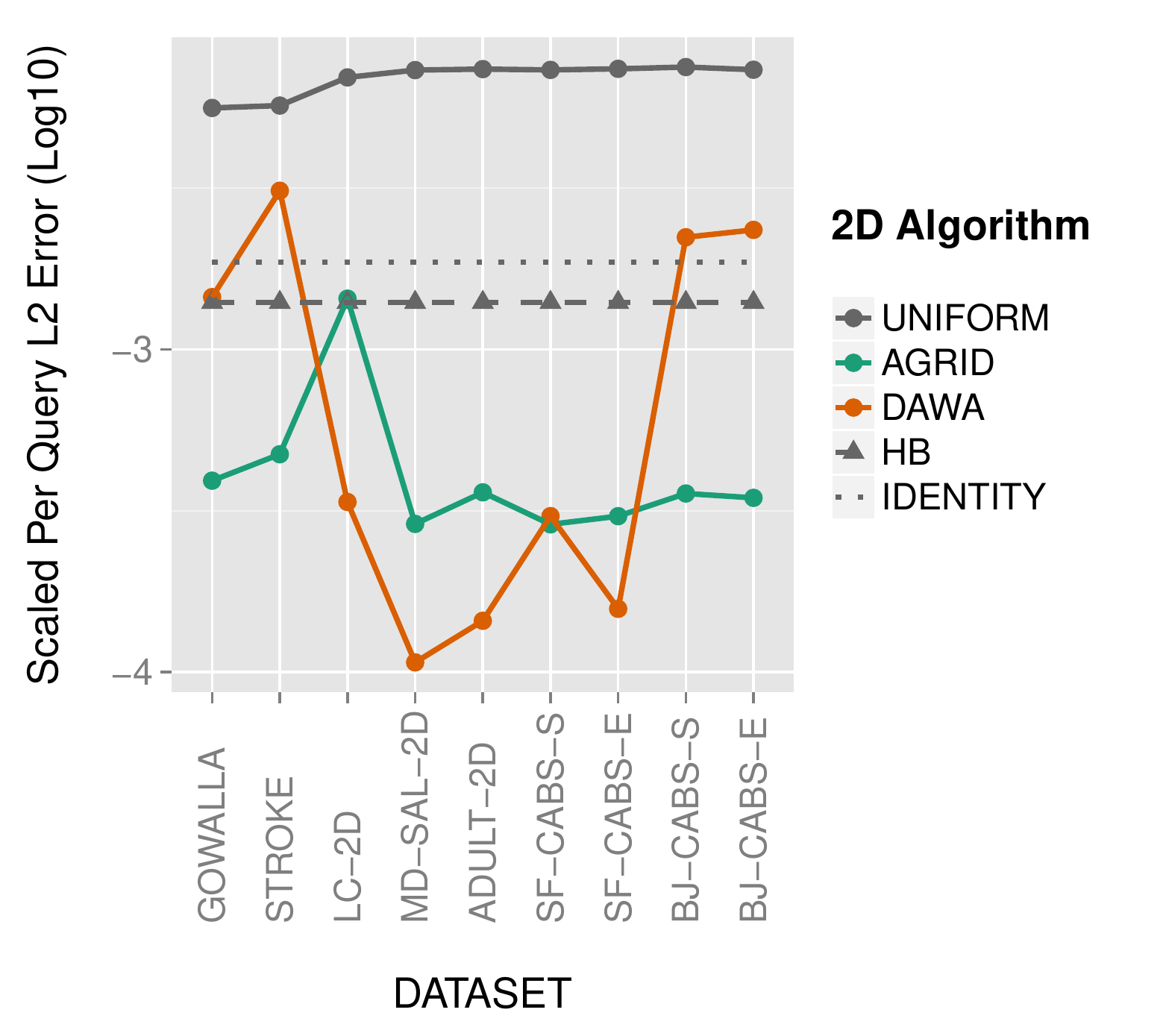}
		\caption{\label{fig:shape2D} Error by shape (2D, scale$=10^4$,domain size$=128\times 128$)}
   \end{subfigure}
   \begin{subfigure}[b]{0.31\textwidth} 	\centering	
		\includegraphics[width=\textwidth]{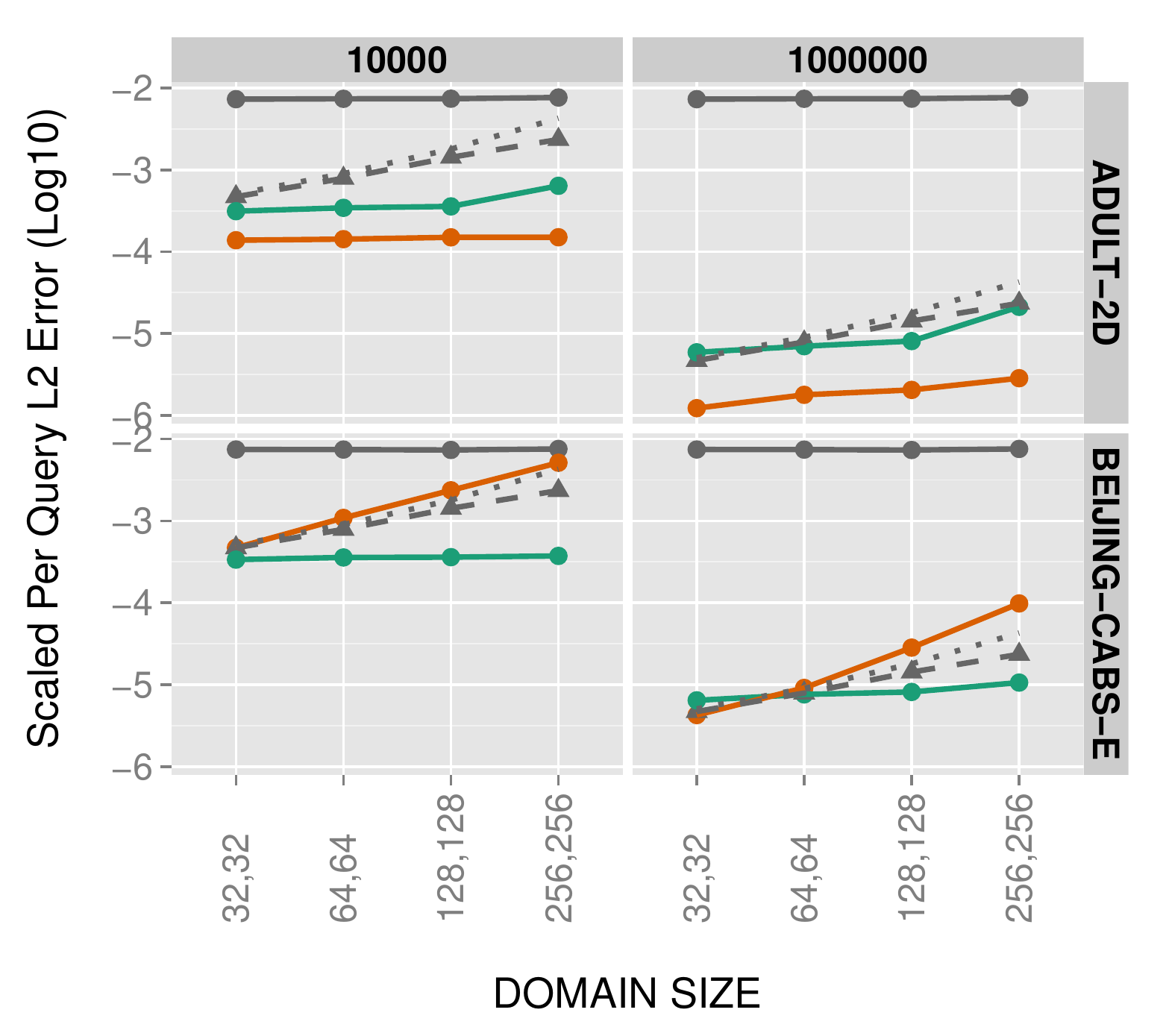}
		\caption{\label{fig:domain2D} Error variation with 2D domain size (scale$\in\{10^4,10^6\}$, shape$\in\{\mbox{\TDadult,\TDbeijingtaxie}\}$)}
   \end{subfigure}
  \caption{ \label{fig:shape-and-domain} Error variation with data shape (1D and 2D cases) and domain size (2D case, for two shapes and two scales).  The only algorithms shown are baselines, \Hierb, and those data-dependent algorithms that were competitive (on some dataset) for the scales shown.}
\end{figure*}

\noindent{\bf Domain Size:} 
\cref{fig:domain2D} shows the effect of domain size for the 2D case.  For two datasets (\TDadult and \TDbeijingtaxie) at two different scales ($10^4$ and $10^6$), it reports error for domain sizes varying from $32 \times 32$ to $256 \times 256$.

\begin{finding}[Different domain size effect]
Domain size affects data-dependent algorithms differently than data-independent algorithms.
\end{finding}

As discussed in~\cite{Qardaji13Understanding}, the error of data-independent algorithms \identity and \Hierb is expected to increase with domain size, and at sufficiently large domains, \Hierb should achieve lower error than \identity.  Those effects are confirmed in \cref{fig:domain2D}.  However, the data-dependent algorithms exhibit different behavior.  \AG's error is nearly flat across domain size.  This makes sense because the grid size is selected independent of domain size.  \dawa's error is flat for some datasets but for others, it increases with domain size.  \dawa selects a partition based on the data shape: for some shapes, the partition selected may grow increasingly fine-grained with domain size.

\begin{table}	[t]
	\begin{subtable}[b]{.49\columnwidth}
	\centering
	{\scriptsize 
	\begin{tabular}{|l|c|c|c|c|}
	\hline
	& \multicolumn{3}{|c|}{Scale} \\ \cline{2-4}
	 &$10^3$	&$10^5$	&$10^7$  \\ \hline
	\dawa		&8	&17	&3  \\
	\Hierb		& 	&14	&16 \\
	\mwemv		&15	& 	&  \\ 
	\efpa		&4	&2	&  \\
	\php		&2	& 	&  \\
	\mwem		&3	& 	&  \\
	\uniform	&3	& 	&  \\
	\ahpv		&	& 	&1 \\
	\hline
			\end{tabular}
	}
	\caption{ \label{tab:competitive-1D} 1D case for domain size $4096$}
	\end{subtable}
	\begin{subtable}[b]{0.49\columnwidth}	
		\centering
	{\scriptsize
	\begin{tabular}{|l|c|c|c|}
	\hline
	& \multicolumn{3}{|c|}{Scale} \\ \cline{2-4}
	 &$10^4$	&$10^6$	&$10^8$\\ \hline
	\dawa		&5	&3	&5 \\
	\AG			&5	&6	&  \\
	\Hierb		& 	& 	&4 \\
	\quadtree   &   &   &4 \\ 
	\ahp		& 	& 	&1 \\
	\hline
	\end{tabular}
	}
	\caption{\label{tab:competitive-2D} 2D case for domain size $128 \times 128$.}
	\end{subtable}
	 \caption{ \label{tab:competitive} Tables show the number of datasets on which algorithms are {\em competitive}, based on statistical significance of mean error.}
\end{table} 

\subsection{Assessment of State of the Art}

\cref{tab:competitive-1D,tab:competitive-2D} show the number of datasets for which an algorithm is {\em competitive} for a given scale. We determine which algorithms are competitive by accounting for the variability in the mean error using t-tests (as discussed in \cref{sec:interpreting}).  

\begin{finding}[Competitive algorithms] \label{finding:utility}
No single algorithm offers uniformly low error. At small scales, data-dependent algorithms dominate; at large scales data-independent algorithms dominate.
\end{finding}

There is no algorithm that dominates all others by offering uniformly lower error.  For 1D, 8 algorithms are competitive for at least one combination of scale and shape; for 2D, 5 algorithms are competitive.
At small scales for both 1D and 2D, the competitive algorithms are all data-dependent.  At the largest scales, the only truly data-dependent algorithms are \dawa and \ahpv (in the 2D case, \quadtree is effectively data-independent at this domain size because the leaves of the tree are individual cells).

While no algorithm universally dominates across settings, \dawa offers the best overall performance in the following sense.  We compare the error experienced by a user who selects a single algorithm to run on all datasets and scales to the error experienced by a user with access to an oracle allowing them to select the optimal algorithm (which would vary among the algorithms mentioned above on a case-by-case basis).  We compute the ratio of errors per dataset/scale and then compute the average ratio, using a geometric mean.  We call this measure ``regret.''  \dawa has a regret of 1.32 on 1D and 1.73 on 2D.  No single algorithm achieves lower regret (the next best on 1D is \Hierb with regret of 1.51 and on 2D it is \AG with regret of 1.90).

\subsection{Tuning Free Parameters}  \label{sec:findings_tuning}

\begin{finding}[Improper Tuning Skews Evaluation] 
For some algorithms tuning free parameters to fit the data can reduce error substantially, in some cases by a factor larger than 7.5.	
\end{finding}
To demonstrate the sensitivity of performance to parameter settings, we ran \ahp, \dawa, and \mwem under several different scenarios (i.e., 1D datasets of varying scale and shape) and found settings of their respective parameters that were optimal for some scenario (i.e., achieved lowest error on some 1D dataset at some scale).
Then for a particular 1D dataset, \medcost,  at a scale of $10^5$, we measure the lowest error and the highest error achieved, considering {\em only those parameter settings that were optimal for some scenario}.  All three algorithms are significantly impacted by the choice of parameters.  Errors can be 2.5x (for \dawa) and around 7.5x (for \mwem and \ahp) larger using sub-optimal parameters (even though those same parameters were in fact optimal for other reasonable inputs).

We proposed a method (\cref{sec:violators}) for setting free parameters in a manner consistent with \cref{pr:param} and for removing dependence on side information (\cref{pr:sideinfo}).  We applied this procedure to \mwem and \ahp -- violators of \cref{pr:param} -- to yield \mwemv and \ahpv respectively. We report our results for the 1D case.

\begin{finding}[Improved \mwem]
Training on synthetic data to decide free parameters (described in \cref{sec:violators}) leads to significant performance improvements for \mwem, especially at large scales.
\end{finding}

\begin{center}
\begin{tabular}{||l|c|c|c|c|c|c||} \hline
\multicolumn{7}{||c||}{Ratio of error ($\mwem / \mwemv$), average all datasets} \\ \hline
Scale & $10^3$ & $10^4$ & $10^5$ & $10^6$ & $10^7$ & $10^8$\\ \hline
 Error ratio & 1.799 & .951 & 1.063 & 5.166 & 12.000 & 27.875 \\ 
\hline
\end{tabular}
\end{center}

The reason for this improved performance is that a stronger signal (here higher scale; equivalently higher epsilon) allows \mwem to benefit from a higher number of rounds because a larger number measurements can be taken with reasonable accuracy.  The training procedure, even with a single synthetic dataset, can help to find the right relationship between the scale-epsilon product and $T$.  Since our method for setting $T$ does not depend on the input data, this result shows that $T$ can be set profitably in a data-independent way.  This result also shows that we can cope with the removal of side information of scale without a significant impact on error. 

For \ahp, we applied our training method to determine parameters $\eta$ and $\propPart$.  We compare the performance of \ahpv against \ahp using one of the several parameters settings used in~\cite{zhangtowards}.
At smaller scales, the error rates of both algorithms were similar.  At larger scales, the results were mixed: one some datasets \ahpv outperformed the fixed setting but on others vice versa.  On average, over all datasets, \ahpv was slightly better, with a ratio of improvement of $1.03$ at scale $10^7$.  This shows that parameter tuning for \ahp is more sensitive to shape and that it may be worth using a portion of the privacy budget to derive parameters.

\subsection{Measurement of Variability and Bias} \label{sec:sub:var}

\begin{finding}[Risk averse algorithm evaluation] {\color{white}.} \\
Algorithms differ in the variation in error.  The algorithm with lowest mean error may not be preferred by a risk averse user. 
\end{finding}

In a real deployment scenario, a user will receive the output of a single run of the algorithm. Thus, she has no opportunity to reliably achieve mean error rates if error rates vary considerably around the mean.  We are not aware of any prior work that has considered this issue in empirical algorithm evaluation. 

While a risk-neutral user may be happy to evaluate algorithms with respect to mean error, a {\em risk averse} user should evaluate algorithms with respect to error rates she is almost certain to achieve.  Therefore, for the risk averse user, we consider the 95$^{th}$ percentile error rate (i.e. the error rate $X$ for which 95\% of trials have error less that $X$).  

In most cases algorithms that are the best according to mean error are also the best according to $95^{th}$ percentile error. In 1D, we see that \dawa has high variability in error and hence in ten scenarios (i.e., 10 scale-shape pairs at domain size 4096), we see \dawa being the best under mean error but not under 95\% error. We see three scenarios (\patent at scale $10^3$, \mdsalary at scale $10^5$ and \nettrace at scale $10^7$, all at domain size 4096) where  an algorithm that was not competitive under mean error has the least 95\% error. That algorithm was either \uniform or \Hierb, both of which have  low variability in error. In 2D we see very few scenarios where an algorithm has the least mean error  but not the least 95\% error.

\begin{finding}[Bias and Consistency]
The high error of \mwem, \mwemv, \php, and \uniform at large scales is due to bias.  These algorithms are proven to be inconsistent.  
\end{finding}

As scale increases, one expects error to decrease.  However, \cref{fig:shape-range} shows that the rate of decrease varies by algorithm.  In fact, for some algorithms -- \php (1D only), \mwem, \mwemv, and \uniform -- error remains quite high even at large scales.  The cause is due to inherent bias introduced by these algorithms.  In empirical results (not shown) we decompose error into bias and variance terms and find that at scale increases, the error of these algorithms becomes dominated by bias.  

We complement our empirical findings with theoretical analysis.  We show these algorithms are, in fact, inconsistent -- implying that the bias does not vanish even as scale (or $\epsilon$) goes to $\infty$.  \cref{tbl:algorithms} reports a complete analysis of consistency for all algorithms; proofs appear in \cref{sec:scale-eps-proofs}.

\subsection{Findings: Reasonable Utility} \label{sec:reasonable_eval}

The previous section compared published techniques in a relative sense: it examined which algorithm performs the best in a given experimental setting.  In this section, we evaluate the algorithms in an absolute sense: do the published algorithms offer reasonable utility?  This is not an easy question to answer, as the notion of ``reasonable'' may be specific to the particular task that an analyst wishes to perform on the dataset.  We address this by identifying situations when the error is unreasonably high for any application by comparing algorithms against a few simple baselines.

\begin{finding}[Comparison to Baselines]
For both 1D and 2D, many algorithms are beaten by the \identity baseline at large scales.  For 1D, the \uniform baseline beats all algorithms on some datasets at small scales.  
\end{finding}

The experiments include two baseline methods (defined in \cref{sec:algorithm_overview}) that can help us assess whether algorithms are providing meaningful improvements.  The first baseline, \uniform learns virtually nothing about the input data (only its scale) and exploits a simplistic assumption of uniformity.  Therefore, we argue that when sophisticated algorithms are delivering error rates comparable to \uniform, even on datasets that are decidedly non-uniform, they are not providing reasonable utility.  \cref{fig:shape-range} shows that most algorithms outperform \uniform at larger scales in 1D and at all scales in 2D.  However, for the 1D case at scale 1000, \uniform achieves lowest error on some datasets.  
By examining \cref{fig:shape1D}, one can see on which datasets \uniform beats competing algorithms.  This figure shows that on some datasets, many algorithms do not offer much improvement over \uniform, suggesting that results for this scale (at least at this $\epsilon$) are unlikely to be useful. 

The second baseline is \identity.  A sophisticated algorithm that does not consistently and significantly improve over \identity does not justify its complexity.  This standard is highly dependent on scale, since, in terms of scaled error, \identity improves with scale.  
For 1D, if we examine \cref{fig:shape-range-1D} and compare mean error across all datasets (the white diamonds), this standard rules out \php, \efpa, \ahpv at moderate scales ($10^5$) and \efpa, \ahpv, \mwem, \mwemv, \php, \dpcube at large scales ($10^7$).  
For 2D (\cref{fig:shape-range-2D}), this standard rules out \mwem, \mwemv, \dpcube, \ahp at moderate scales ($10^6$) and \AG, \mwem, \mwemv, \dpcube, \ahp at large scales ($10^8$).

%!TEX root = paper.tex
\vspace{-1em}
\section{Discussion and Takeaways} \label{sec:discussion}
In this section, we synthesize our findings to (a) explain gaps or resolve inconsistencies from prior work, (b) present lessons for practitioners and (c) pose open research questions. 

\stitle{Explaining Results from  Prior Work}\label{sec:discussion-priorwork}
We revisit the gaps and inconsistencies identified in \cref{sec:prior_results} in light of our findings.

\mybullet We systematically compared data-dependent algorithms against the state of the art data-independent ones and learn there is no clear winner. Both scale and shape significantly affect algorithm error. Moreover, shape affects each data dependent algorithm differently.

\mybullet We can resolve the apparent inconsistency between \cite{hardt2012a-simple} and \cite{Li14Data-} regarding the comparative performance of \mwem against data-independent matrix mechanism techniques by accounting for the effect of scale as well as the scale-epsilon exchangeability property.  \mwem outperforms data-independent techniques (including instances of the matrix mechanism) either when the scale is small or when $\epsilon$ is small.  Hardt et al. evaluate on small datasets and low $\epsilon$ whereas  the datasets of Li et al. where \mwem struggled turn out to be datasets with large scale (over $335,000$).  We demonstrated this by using the same datasets as \cite{hardt2012a-simple, Li14Data-} and controlling for scale.

\mybullet Our results are consistent with the findings of Qardaji et al.~\cite{Qardaji13Understanding} regarding domain size and data-independent hierarchical algorithms.  Interestingly, however, we show that data-dependent algorithms have a different relationship with domain size.

\mybullet We conduct a comprehensive evaluation of algorithms for answering 2D range queries, which is missing from the literature.

\stitle{Lessons for practitioners}
Our results identify guidelines for practitioners seeking the best utility for their task and facing the daunting task of selecting and configuring an appropriate algorithm.

The first consideration for a practitioner should be the overall strength of the "signal" available to them.  This is determined by both the $\epsilon$ budget and the scale of the dataset.  (We have shown that these two factors are exactly exchangeable in their impact on scaled error.)  In a ``high signal'' regime (high scale, high $\epsilon$), it is unlikely that any of the more complex, data-dependent algorithms will beat the simpler, easier-to-deploy, data independent methods such as \identity and \Hierb.  This greatly simplifies algorithm selection and deployment because error bounds are easy to derive, performance doesn't depend on the input dataset, and there are few parameters to set for these algorithms.  In ``low signal'' regimes (low scale, low $\epsilon$), deploying a data-dependent algorithm should be seriously considered, but with an awareness of the limitations: depending on the properties of the input data, error can vary considerably and error bounds are not provided by the data-dependent algorithms. For answering 1D range queries, \dawa is a competitive choice on most datasets in the low and medium signal regimes. For 2D range queries, \AG consistently beats the data independent techniques, but \dawa can significantly outperform \AG and the data independent techniques on very sparse datasets. 

\stitle{Open research Problems}

Our experimental findings raise a number of research challenges we believe are important for advancing differentially private algorithms.   
\mybullet {\em Understanding Data Dependence:} We have shown that data-dependent algorithms do not appear to be exploiting the same features of shape.  The research community appears to know very little about the features of the input data that permit low error.  Are there general algorithmic techniques that can be effective across diverse datasets, or should algorithm designers seek a set of specialized approaches that are effective for specific data properties?

\mybullet {\em Algorithm Selection:} While our evaluation of 1D and 2D algorithms identifies a class of state-of-the-art algorithms, it still does not fully solve the problem faced by the practitioner of selecting the algorithm that would result in the least error given a new dataset. We believe further research into analytical and empirical methods for algorithm selection would greatly aid the adoption of differential privacy algorithms in real world systems.

\mybullet {\em Error Bounds:} A related problem is that data-dependent algorithms typically do not provide public error bounds (unlike, e.g., the Laplace mechanism).   Hence, users cannot predict the error they will witness without knowing the dataset, presenting a major challenge for algorithm deployment.  Developing publishable error bounds is particularly important for this class of algorithms.

\mybullet {\em Parameter Tuning:} We showed that  parameter tuning can result in significant gains in utility. Our training procedure is a first step towards setting parameters, but already distinguishes between a case where simple tuning can greatly improve performance over a fixed value (\mwem) and a case where parameter setting is more challenging because of data-dependence (\ahp).

\mybullet {\em Expected Error versus Variability:} Existing empirical evaluations have focused on mean error while ignoring the variation of error over trials.  We show that risk-averse users favoring low variability over low expected error may choose different algorithms than risk-seeking users.  Current algorithms do not offer users the ability to trade off expected error and variability of achieved error, which could be an important feature in practice.

%!TEX root = paper.tex
\section{Conclusion}\label{sec:conclusion}

We have presented \dpbench a novel and principled  framework for evaluating differential privacy algorithms. We use \dpbench to evaluate algorithms for answering 1- and 2-D range queries and as a result (a) resolved gaps/inconsistencies in prior work, (b) identified state-of-the-art algorithms that  achieve the least error for their datasets, and (c) posed open research questions. 

We are eager to extend our investigation in a number of ways.  We have focused here primarily on evaluating the utility of algorithms, assuming that standard settings of $\epsilon$ offer sufficient privacy.  We would like to provide meaningful guidelines for setting privacy parameters along with utility standards.  We hope to expand our investigation to broader tasks (beyond 1- and 2-D range queries), and, as  noted in \cref{tbl:algorithms}, some of the algorithms considered support broader classes of workloads and may offer advantages not seen in our present experiments.

\noindent
\textbf{Acknowledgments} We appreciate the comments of each of the anonymous reviewers. This material is based upon work supported by the National Science Foundation under Grant Nos.~1253327, 1408982, 1443014, 1409125, and 1409143.

{\small
\bibliographystyle{abbrv}
\bibliography{bib/refs}

\begin{thebibliography}{10}

\bibitem{Acs2012compression}
G.~{\'A}cs, C.~Castelluccia, and R.~Chen.
\newblock Differentially private histogram publishing through lossy
  compression.
\newblock In {\em ICDM}, pages 1--10, 2012.

\bibitem{blum2013learning}
A.~Blum, K.~Ligett, and A.~Roth.
\newblock A learning theory approach to noninteractive database privacy.
\newblock {\em Journal of the ACM (JACM)}, 60(2):12, 2013.

\bibitem{chaudhuri2013a-stability-based}
K.~Chaudhuri and S.~A. Vinterbo.
\newblock A stability-based validation procedure for differentially private
  machine learning.
\newblock In {\em Advances in Neural Information Processing Systems}, pages
  2652--2660, 2013.

\bibitem{Cormode11Differentially}
G.~Cormode, M.~Procopiuc, E.~Shen, D.~Srivastava, and T.~Yu.
\newblock Differentially private spatial decompositions.
\newblock In {\em ICDE}, pages 20--31, 2012.

\bibitem{Dwork08Differential}
C.~Dwork.
\newblock Differential privacy: A survey of results.
\newblock In {\em TAMC}, 2008.

\bibitem{dwork2011a-firm}
C.~Dwork.
\newblock A firm foundation for private data analysis.
\newblock {\em Communications of the {ACM}}, 54(1):86--95, 2011.

\bibitem{dwork2006calibrating}
C.~Dwork, F.~M.~K. Nissim, and A.~Smith.
\newblock Calibrating noise to sensitivity in private data analysis.
\newblock In {\em TCC}, pages 265--284, 2006.

\bibitem{Dwork14Algorithmic}
C.~Dwork and A.~Roth.
\newblock {\em The Algorithmic Foundations of Differential Privacy}.
\newblock Found. and Trends in Theoretical Computer Science, 2014.

\bibitem{fredrikson14usenix}
M.~Fredrikson, E.~Lantz, S.~Jha, S.~Lin, D.~Page, and T.~Ristenpart.
\newblock Privacy in pharmacogenetics: An end-to-end case study of personalized
  warfarin dosing.
\newblock In {\em USENIX Security}, 2014.

\bibitem{hardt2012a-simple}
M.~Hardt, K.~Ligett, and F.~McSherry.
\newblock A simple and practical algorithm for differentially private data
  release.
\newblock In {\em NIPS}, 2012.

\bibitem{hay2010boosting}
M.~Hay, V.~Rastogi, G.~Miklau, and D.~Suciu.
\newblock Boosting the accuracy of differentially private histograms through
  consistency.
\newblock {\em PVLDB}, 3(1-2):1021--1032, 2010.

\bibitem{he2015dpt}
X.~He, G.~Cormode, A.~Machanavajjhala, C.~M. Procopiuc, and D.~Srivastava.
\newblock Dpt: differentially private trajectory synthesis using hierarchical
  reference systems.
\newblock {\em Proceedings of the VLDB Endowment}, 8(11):1154--1165, 2015.

\bibitem{Hu2015Telco}
X.~Hu, M.~Yuan, J.~Yao, Y.~Deng, L.~Chen, Q.~Yang, H.~Guan, and J.~Zeng.
\newblock Differential privacy in telco big data platform.
\newblock {\em Proc. VLDB Endow.}, 8(12):1692--1703, Aug. 2015.

\bibitem{Jagadish:1998:OHQ:645924.671191}
H.~V. Jagadish, N.~Koudas, S.~Muthukrishnan, V.~Poosala, K.~C. Sevcik, and
  T.~Suel.
\newblock Optimal histograms with quality guarantees.
\newblock In {\em VLDB}, pages 275--286, 1998.

\bibitem{Li14Data-}
C.~Li, M.~Hay, and G.~Miklau.
\newblock A data- and workload-aware algorithm for range queries under
  differential privacy.
\newblock {\em PVLDB}, 2014.

\bibitem{Li:2010Optimizing-Linear}
C.~Li, M.~Hay, V.~Rastogi, G.~Miklau, and A.~McGregor.
\newblock Optimizing linear counting queries under differential privacy.
\newblock In {\em PODS}, pages 123--134, 2010.

\bibitem{chaopvldb12}
C.~Li and G.~Miklau.
\newblock An adaptive mechanism for accurate query answering under differential
  privacy.
\newblock {\em PVLDB}, 5(6):514--525, 2012.

\bibitem{li2015matrix}
C.~Li, G.~Miklau, M.~Hay, A.~McGregor, and V.~Rastogi.
\newblock The matrix mechanism: optimizing linear counting queries under
  differential privacy.
\newblock {\em The VLDB Journal}, pages 1--25, 2015.

\bibitem{mcsherry2009pinq}
F.~D. McSherry.
\newblock Privacy integrated queries: an extensible platform for
  privacy-preserving data analysis.
\newblock In {\em SIGMOD}, pages 19--30, 2009.

\bibitem{epfl-mobility-20090224}
M.~Piorkowski, N.~Sarafijanovic-Djukic, and M.~Grossglauser.
\newblock {CRAWDAD} dataset epfl/mobility (v. 2009-02-24).
\newblock Downloaded from http://crawdad.org/epfl/mobility/20090224, Feb. 2009.

\bibitem{qardaji2013differentially}
W.~Qardaji, W.~Yang, and N.~Li.
\newblock Differentially private grids for geospatial data.
\newblock In {\em Data Engineering (ICDE), 2013 IEEE 29th International
  Conference on}, pages 757--768. IEEE, 2013.

\bibitem{Qardaji13Understanding}
W.~Qardaji, W.~Yang, and N.~Li.
\newblock Understanding hierarchical methods for differentially private
  histograms.
\newblock {\em PVLDB}, 6(14), 2013.

\bibitem{sandercock2011international}
P.~A. Sandercock, M.~Niewada, A.~Cz{\l}onkowska, et~al.
\newblock The international stroke trial database.
\newblock {\em Trials}, 12(1):1--7, 2011.

\bibitem{xiao2011ireduct}
X.~Xiao, G.~Bender, M.~Hay, and J.~Gehrke.
\newblock ireduct: Differential privacy with reduced relative errors.
\newblock In {\em SIGMOD}, 2011.

\bibitem{xiao2010differential}
X.~Xiao, G.~Wang, and J.~Gehrke.
\newblock Differential privacy via wavelet transforms.
\newblock In {\em ICDE}, pages 225--236, 2010.

\bibitem{xiao2014dpcube}
Y.~Xiao, L.~Xiong, L.~Fan, S.~Goryczka, and H.~Li.
\newblock {DPCube}: Differentially private histogram release through
  multidimensional partitioning.
\newblock {\em Transactions of Data Privacy}, 7(3), 2014.

\bibitem{xu2013differential}
J.~Xu, Z.~Zhang, X.~Xiao, Y.~Yang, G.~Yu, and M.~Winslett.
\newblock Differentially private histogram publication.
\newblock {\em The VLDB Journal}, pages 1--26, 2013.

\bibitem{Yuan12Low-Rank}
G.~Yuan, Z.~Zhang, M.~Winslett, X.~Xiao, Y.~Yang, and Z.~Hao.
\newblock Low-rank mechanism: Optimizing batch queries under differential
  privacy.
\newblock {\em PVLDB}, 5(11):1136--1147, 2012.

\bibitem{zhangtowards}
X.~Zhang, R.~Chen, J.~Xu, X.~Meng, and Y.~Xie.
\newblock Towards accurate histogram publication under differential privacy.
\newblock In {\em ICDM}, 2014.

\end{thebibliography}
}

\appendix
%!TEX root = paper.tex
\section{Dataset Descriptions} \label{sec:dataset-desc}

\cref{tbl:datasets} provides an overview of all datasets, 1D and 2D, considered in the paper.
\stitle{1D Datasets}

The first seven were described in the papers in which they originally appeared \cite{hardt2012a-simple,Li14Data-,hay2010boosting,Acs2012compression,zhangtowards,xu2013differential}.  The eleven new datasets were derived from three original data sources.  A single primary attribute was selected for an initial shape distribution (denote by suffix {\sc -ALL}).  In addition, filters on a secondary attribute were applied to the data, resulting in alternative 1D shape distributions, still on the primary attribute.

The {\sc Bids}\xspace datasets are derived from a Kaggle competition whose goal is to identify online auctions that are placed by robots.\footnote{https://www.kaggle.com/c/facebook-recruiting-iv-human-or-bot/data} The histogram attribute is the IP address of each individual bid. \bidsfj and \bidsfm result from applying distinct filtering conditions on attribute ``merchandise'': \bidsfj is a histogram on IP address but counting only individuals where ``merchandise=jewelry'', \bidsfm is a histogram on IP address but counting only individuals where ``merchandise=mobile''.

The {\sc MD-SAL}\xspace datasets are based on the Maryland salary database of state employees in 2012.\footnote{http://data.baltimoresun.com/salaries/state/cy2012/}  The histogram attribute is ``YTD-gross-compensation''. \mdsalaryfa is filtered on condition ``pay-type=Annually''. 

The {\sc LC}\xspace datasets are derived from data published by the ``Lending Club'' online credit market.\footnote{https://www.lendingclub.com/info/download-data.action}  It describes loan applications that were rejected. The {\sc LC-REQ}\xspace datasets are based on attribute ``Amount Requested'' while the {\sc LC-DTIR}\xspace datasets are based on attribute ``Debt-To-Income Ratio''.  In both cases, additional shapes are generated by applying filters on the ``Employment'' attribute. \lcreqfa and \lcdtirfa only count records with attribute ``Employment'' in range $\left[0,5\right]$, \lcreqfb and \lcdtirfb only count records with attribute ``Employment'' in range $\left(5,10\right]$.

\stitle{2D Datasets}
The {\sc BJ-CABS} datasets and \TDcheckin have been used and described in previous papers\cite{he2015dpt,qardaji2013differentially}. The {\sc SF-CABS} datasets are derived from well-known mobility traces data from San Francisco taxis \cite{epfl-mobility-20090224}.These location-based data are represented as latitude longitude pairs. To get more diverse data shape, we divide the two cab trace data into four by using only the start point and end point in a single dataset. \TDbeijingtaxis and \TDcabspottings contain only the start location of a cab trip, while \TDbeijingtaxie and \TDcabspottinge record only the end locations.

\TDadult, \TDmdsalary and \TDloan are derived from same sources as a subset of data we used in the 1D experiments. \TDadult is derived from source of \adult. This source is also used in \cite{hardt2012a-simple} with attributes ``age'' and ``hours'', but  we use the attributes ``capital-gain'' and ``capital-loss'' in order to test on a larger domain. \TDmdsalary is based on the Maryland salary database using the attributes ``Annual Salary'' and ``Overtime earnings''. \TDloan is derived from the accepted loan data from Lending Club with attributes ``Funded Amount'' and ``Annual Income''.

 \TDstoke is a new dataset derived from the International Stroke Trial database\cite{sandercock2011international}, which is one of the largest randomized trial conducted in acute stroke for public use. We used the attributes ``Age'' and ``Systolic blood pressure''.

%!TEX root = paper.tex

\section{Algorithm Descriptions} \label{sec:alg-desc}

\cref{sec:algorithm_overview} and \cref{tbl:algorithms} provide an overview of the algorithms studied; here we describe each algorithm individually.  We begin with data-independent approaches, whose expected error rate is the same for all datasets, then describe data-dependent approaches. Baseline algorithms \identity and \uniform are already described in \cref{sec:algorithm_overview}.

\stitle{Data-Independent Algorithms}
Each of the data-independent approaches we consider can be described as an instance of the matrix mechanism \cite{Li:2010Optimizing-Linear,li2015matrix} (although a number were developed independently).  The central idea is to select a set of linear queries (called the strategy), estimate them privately using the Laplace mechanism, and then use the noisy results to reconstruct answers to the workload queries.  The strategy queries, which can be conveniently represented as a matrix, should have lower sensitivity than the workload and should allow for effective reconstruction of the workload.  

If the discrete Haar wavelet matrix is selected as the strategy, the result is the \privelet method~\cite{xiao2010differential} which is based on the insight that any range query can be reconstructed by just a few of the wavelet queries and that the sensitivity of the Haar wavelet grows with $\prod_{i=1}^k \log_2 n_i$ where $k$ is the dimensionality (number of attributes) and $n_i$ is the domain size of the $i^{th}$ attribute.  

Several approaches have been developed that use a strategy consisting of hierarchically structured range queries.  Conceptually the queries can be arranged in a tree.  The leaves of the tree consist of individual  queries $x_i$ (for $i = 1, \dots, n$).  Each internal node computes the sum of its children; at the root, the sum is therefore equal to the number of records in the database, $\scale{\x}$.  
The approaches differ in terms of the branching factor and the privacy budget allocation.  
The \Hier method~\cite{hay2010boosting} has a branching factor of $b$ and uniform budget allocation.  The \Hierb method~\cite{Qardaji13Understanding} uses the domain size to determine the branching factor: specifically, it finds the $b$ that minimizes the average variance of answering all range queries through summations of the noisy counts in the tree.  It also uses a uniform budget allocation.  \greedyH is a subroutine of the \dawa~\cite{Li14Data-} algorithm but can also be used as a stand-alone algorithm.  It has a branching factor $b$ and employs a greedy algorithm to tune the budget allocation to the workload.  When used in \dawa, \greedyH is applied to the partition that \dawa computes; by itself, it is applied directly to $\x$.

All of the above techniques reduce error by using a set of strategy queries which are a good match for the workload of interest.  \Hier, \Hierb, and \privelet were all initially designed to answer the set of {\em all} range queries, but the strategy each approach uses is static -- it is not adapted to the particular queries in the workload.  \greedyH is the only technique that modifies the strategy (by choosing weights for hierarchical queries) in response to the input workload.  In the table, it is marked as ``workload-aware.''  While \Hier and \Hierb naturally extend to the multi-dimensional setting (quadtrees, octrees, etc.), \greedyH extends to 2D by applying a Hilbert transform to obtain a 1D representation. 

\stitle{Data-Dependent Partitioning Algorithms} 

Many of the data-dependent approaches use a partitioning step to approximate the input data vector and we discuss them as a group.  These algorithms partition the domain into a collection of disjoint buckets whose union span the domain.
The partition must be selected using a differentially private algorithm, since it is highly dependent on the input data, and this constitutes the first step for each algorithm in this group.  Once a partition is selected, in a second step, they obtain noisy counts only for the buckets. 
A noisy data vector can be derived from the noisy bucket counts by assuming the data is uniform within each bucket, and this vector is then used to answer queries.  Partitioning trades off between two sources of error: the error due to adding noise to the bucket count, which diminishes with increased bucket size, and the error from assuming uniformity within a bucket, which increases when non-uniform regions are lumped together.  When the partition step can effectively find regions of the data that are close to uniform, these algorithms can perform well. 

By the sequential composition property of differential privacy, any $\epsilon_1$-differentially private algorithm for partition selection can be combined with any $\epsilon_2$-differentially private algorithm for count estimation and achieve $\epsilon$-differential privacy provided that $\epsilon_1 + \epsilon_2 \leq \epsilon$.  All these algorithms therefore share a parameter that determines how to allocate the privacy budget across these two steps.  We will use $\propPart$ to denote the proportion of $\epsilon$ used to identify the partition.  Thus, $\epsilon_1 = \epsilon \propPart$ and $\epsilon_2 = \epsilon (1 - \propPart)$.

Unlike the first step, the approaches for estimating the bucket counts are typically data-{\em in}dependent. In fact, \php, \ahp, \structurefirst, and \UG use the Laplace mechanism to obtain noisy counts.

The distinguishing features of each algorithm in this group are:

\php~\cite{Acs2012compression} finds a partition by recursively bisecting the domain into subintervals.  The midpoint is found using the exponential mechanism with a cost function based on the expected absolute error from the resulting partition.  \php is limited to 1D.

\ahp~\cite{zhangtowards} uses the Laplace mechanism to obtain a noisy count for each $x_i$ for $i = 1, \dots, n$ and sets noisy counts below a threshold to zero.  The counts are then sorted and clustered to form a partition.  As was done by \cite{zhangtowards}, the greedy clustering algorithm is used in the experiments.  The threshold is determined by parameter $\eta$.  \ahp extends to the multi-dimensional setting.

\dawa~\cite{chaopvldb12} uses dynamic programming to compute the least cost partition in a manner similar in spirit to V-optimal histograms~\cite{Jagadish:1998:OHQ:645924.671191}.  The cost function is the same as the one used in the exponential mechanism of \php~\cite{Acs2012compression}.  To ensure differential privacy, noisy costs are used in place of actual costs.  Once the partition is found, a hierarchical set of strategy queries are derived using \greedyH.  To operate on 2D data, \dawa applies a Hilbert transformation.

\structurefirst~\cite{xu2013differential} aims to find a partition that minimizes the expected sum of squared error.  While the aforementioned algorithms choose the size of the partition privately, in \structurefirst, the number of buckets is specified as a parameter $k$.  Xu et al.~\cite{xu2013differential} recommend setting $k$ as a function of the domain size and we follow those guidelines in our experimental evaluation (see \cref{sec:impl_details}).  Given the specified number of buckets, the bucket boundaries are selected privately using the exponential mechanism.  Xu et al.~\cite{xu2013differential} propose two variants of the algorithm, one based on the mean and another based on the median; our experiments evaluate the former.  
The parameters of \structurefirst include $k$, the number of buckets; $\propPart$, the ratio for privacy budget allocation; and $F$, an upper bound on the count of a bucket (which is necessary for the cost function used in the exponential mechanism to select bucket boundaries).  \structurefirst is limited to 1D.

\dpcube~\cite{xiao2014dpcube}, which is multi-dimensional, selects the partition by first applying the Laplace mechanism to obtain noisy counts for each cell in the domain, and then running a standard kd-tree on the noisy counts.  Once the partition is selected it obtains fresh noisy counts for the partitions and uses inference to average the two sets of counts.

The remaining partitioning algorithms 
are designed specifically for 2D.

\quadtree~\cite{Cormode11Differentially} generates a quadtree with fixed height and then noisily computes the counts of each node and does post-processing to maintain consistency.  The height of the tree is a parameter.  This algorithm becomes data-dependent if the tree is not high enough (i.e., the leaves contains aggregations of individual cells).  Because the partition structure is fixed, this algorithm does not use any privacy budget to select it (i.e., $\propPart = 0$).

\hybridtree~\cite{Cormode11Differentially} is the combination of a kd-tree and the aforementioned quadtree. It uses a differentially private algorithm to build a kd-tree in a top down manner (\dpcube goes bottom up).  This tree extends a few levels and then a fixed quadtree structured is used for the remaining levels until a pre-specified height is reached.

\UG~\cite{qardaji2013differentially} builds an equi-width partition where the width is chosen in a data-dependent way, based on the scale of the dataset.  \AG~\cite{qardaji2013differentially} builds a two-level hierarchy of partitions.  The top level partition produces equi-width buckets, similar to \UG.  Then within each bucket of the top-level partition, a second partition is chosen in a data-adaptive way based on a noisy count of the number of records in the bucket. 

\stitle{Other Data-Dependent Algorithms}

\mwem~\cite{hardt2012a-simple} is a workload aware algorithm that supports arbitrary workloads of linear queries (not just range queries).  It starts with an estimate of the dataset that is completely uniform (based on assumed knowledge of the scale), and updates the estimate iteratively.  Each iteration privately selects the workload query with highest error and then updates the estimate using a multiplicative weights update step.  The number of iterations is an important parameter whose setting has a significant impact on performance.   We propose a method for setting $T$, described in \cref{sec:violators}, that adheres to our evaluation principles.  
 
\efpa~\cite{Acs2012compression} is based on applying the discrete Fourier transform (DFT) to a 1D data vector $\x$.  It retains the top $k$ Fourier coefficients, adds noise to them using the Laplace mechanism, and then inverts the transformation.  By choosing $k < n$, the magnitude of the noise is lowered at the potential cost of introducing approximation error.  The value of $k$ is chosen in a data-dependent way using the exponential mechanism with a cost function that is based on expected squared error.  The privacy budget is divided evenly selecting $k$ and measuring the $k$ coefficients.

%!TEX root=paper.tex

\section{Theoretical Analysis} \label{sec:scale-eps-proofs}
In this section, we theoretically analyze whether the algorithms considered in the paper satisfy scale-epsilon exchangeability (\cref{def:exchangeability}) and consistency (\cref{def:consistency}).
\begin{lemma}
\label{lemma:1}
Any instance of the Matrix Mechanism satisfies consistency and scale-epsilon exchangeability.
\end{lemma}
\begin{proof}
If Matrix Mechanism uses strategy $\mathbf{S}$ to answer workload $\W$ on data vector $\x$, the scaled per query error ($err$) is
\begin{eqnarray*}
err &=& \frac{1}{\scale{\x}\cdot |\W|}||\W\mathbf{S}^{-1}(\mathbf{S}\x + Lap(\frac{\Delta}{\epsilon})) - \W\x||_2 \\
&=& \frac{1}{\scale{\x}\cdot |\W|}||\W\mathbf{S}^{-1}Lap(\frac{\Delta}{\epsilon})||_2 
= \alpha * \frac{1}{\epsilon \scale{\x}}
\end{eqnarray*}
where $\alpha$ is constant in $\epsilon$ and $\scale{\x}$. Thus, Matrix Mechanism ensures both consistency and scale-epsilon exchangeability.
\end{proof}

\begin{theorem} \label{thm:dataind_consistency}
\identity, \privelet, \Hier, \Hierb, \greedyH ensure both consistency and scale-epsilon exchangeability.
\end{theorem}
\begin{proof}
All these five data independent methods are the instances of Matrix Mechanism. \cref{lemma:1} shows Matrix Mechanism satisfies consistency and scale-epsilon exchangeability. Thus, they all ensure both consistency and scale-epsilon exchangeability.
\end{proof}

\begin{lemma}
\label{lemma:2}
When $\epsilon$ goes to infinity, Exponential Mechanism picks one of the items with the highest score with probability 1.
\end{lemma}
\begin{proof}
Suppose $T$ is the set of all items, $T^*$ contains all the items with the highest score $\alpha$. Since Exponential Mechanism picks any item t with the probability proportional to $e^{\epsilon score(t)}$, the probability of picking any item in $T^*$ is 
\begin{eqnarray*}
P(T^{*}) &=& \frac{\sum_{t \in T^*}e^{\epsilon \alpha}}{\sum_{t \in T}e^{\epsilon score(t)}} 
= \frac{\sum_{t \in T^*}e^{\epsilon \alpha}}{\sum_{t \in T^*}e^{\epsilon \alpha} + \sum_{t \in T - T^*}e^{\epsilon score(t)}} \\
&=&\frac{|T^*|}{|T^*| + \sum_{t \in T - T^*}e^{\epsilon (score(t) - \alpha)}} 
=  1 (as \epsilon \to \infty)
\end{eqnarray*} 
\end{proof}

\begin{theorem}
\efpa ensures consistency.
\end{theorem}
\begin{proof}
\efpa works in two steps: (a) pick $k<n$ coefficients, and (b) add noise to the chosen coefficients. The number of coefficients $k$ is chosen by Exponential Mechanism with the score function of expected noise added to the $k$ chosen coefficients ($\frac{2k}{\epsilon}$) plus the error due to dropping $n-k$ Fourier coefficients. When $\epsilon$ goes to infinity, the noise term goes to zero. Thus $k = n$ (where all the coefficients are chosen) receives the highest score. By \cref{lemma:2}, this will be chosen as $\epsilon$ tends to infinity. Therefore, the scaled error from \efpa will tend to zero when $\epsilon$ tends to infinity, which satisfies consistency.
\end{proof}

\begin{theorem}
\ahp, \dawa and \dpcube are consistent.
\end{theorem}
\begin{proof}
All \ahp, \dawa and \dpcube noisily partition the domain and then estimate partition counts. When $\epsilon$ tends to infinity, these algorithms will use a partitioning with zero bias. Also, partition counts are estimated using Laplace Mechanism whose scaled error goes to zero as $\epsilon$ or scale tends to infinity. Thus \ahp, \dawa and \dpcube ensure consistency.
\end{proof}

\begin{theorem}
\AG and \UG ensure consistency.
\end{theorem}
\begin{proof}
When $\epsilon$ goes to infinity, both \AG and \UG are equal to \identity having grids with size 1. Thus \AG and \UG are consistent.
\end{proof}

\begin{theorem}
On sufficiently large domains, \quadtree and \hybridtree do not ensure consistency.
\end{theorem}
\begin{proof}
Both \quadtree and \hybridtree set the maximum height of the tree. If the input dataset has large domain size, the leaf node in the tree contains more than one cell in the domain introducing bias. Therefore, \quadtree and \hybridtree are not consistent.
\end{proof}

\begin{theorem}
PHP does not ensure consistency.
\end{theorem}
\begin{proof}
PHP recursively bisects the domain but sets the maximum iterations to be $\log_2 n$ (n is the domain size). After $\log_2 n$ iterations, it is possible that there exist clusters with bias greater than zero.
\begin{example}
Let the scale $m = 2^n - 1$ and $X = \{x_1,\dots,x_n\}$ be the dataset with $x_i = 2^{n-i}$. In iteration j $(\epsilon = \infty)$, PHP will divide the domain into $\{x_j\}$ and $\{x_{j+1},\dots,x_n\}$. After $\log_2 n$ iterations, the last partition will contain domain elements with different counts. Thus the bias of the partition does not go to zero even when $\epsilon$ tends to infinity.
\end{example}
\end{proof}

The consistency of \structurefirst depends on the implementation.  Our experiments use the implementation provided by the authors that includes the modification described in Sec.~6.2 of~\cite{xu2013differential}.
\begin{theorem}
StructureFirst does not satisfy consistency.  If the modification described in Sec.~6.2 of~\cite{xu2013differential} is applied, then it does satisfy consistency.
\end{theorem}
\begin{proof}
StructureFirst fixes the number of clusters first (it sets the number k to be $\frac{n}{10}$ empirically, n is the domain size), which will introduce bias if the real number of clusters are greater than k. 
\begin{example}
Let $X = \{x_1,\dots,x_n\}$ be the dataset with $x_i = i$. If we set the number of clusters $k<n$ (suppose $k = \frac{n}{10}$), there will be one cluster $c$ containing more than one domain point. Since all the $x_i$ are different, the bias of $c$ does not go to zero even if $\epsilon$ tends to infinity.
\end{example}
The modification proposed in~\cite{xu2013differential} builds a hierarchical histogram (i.e., \Hier) within each bin.  This makes the algorithm consistent for the same reason that \Hier is consistent.
\end{proof}

\begin{theorem} \label{thm:mwem_consistency}
MWEM does not ensure consistency.
\end{theorem}
\begin{proof}
MWEM fixes the number of iterations (it sets T = 10 empirically), which will introduce bias if T queries are not enough to update the estimate.
\begin{example}
Let $X = \{x_1,\dots,x_n\}$ be the dataset with $x_i = i$ and workload W contains all the single count queries. If the number of iterations T is less than n (suppose T = 10 $<$ n), we can only pick T single count queries and correct T domain point when $\epsilon$ goes to infinity. The bias generated from other domain points will not go to zero even if $\epsilon$ tends to infinity.
\end{example}
\end{proof}

\begin{lemma}
\label{lemma:3}
Laplace Mechanism is scale-epsilon exchangeable. Exponential Mechanism is scale-epsilon exchangeable when the score is a linear function of scale.
\end{lemma}
\begin{proof}
Laplace Mechanism ensures scale-epsilon exchangeability since the scaled expected $L_1$ error for Laplace Mechanism is $\frac{\Delta}{\epsilon \Lone{\x}}$ ($\Delta$ is the fixed sensitivity given workloads). 

Exponential Mechanism picks any item t with the probability proportional to $e^{\epsilon score(t)}$. When the score is a linear function of scale, Exponential Mechanism is scale-epsilon exchangeable.
\end{proof}

\begin{theorem}
PHP, MWEM, EFPA are scale-epsilon exchangeable.
\end{theorem}
\begin{proof}
All these three methods use Exponential Mechanism to choose clustering or parameter. Since their score functions are all linear functions of scale, based on \cref{lemma:3}, they will choose the same clustering or parameter with the same probability when $\epsilon$ * scale is fixed. Once they use the same clustering or parameter, the scaled error is the sum of the bias and the variance. The variance comes from Laplace Mechanism which is scale-epsilon exchangeable by \cref{lemma:3}. It is also easy to check those bias remain the same when fixing $\epsilon$ * scale. Thus PHP, MWEM, EFPA ensure scale-epsilon exchangeability.
\end{proof}

\begin{theorem}
StructureFirst does not satisfy scale-epsilon exchangeability.
\end{theorem}
\begin{proof}
The score function used in StructureFirst is a linear function of the square of scale. Based on \cref{lemma:3}, even if $\epsilon$ * scale is fixed, the probability of picking the same clustering will be different for different $\epsilon$ settings. Once the clustering is different, the bias will be different.
\begin{example}
Suppose we use $\epsilon = 1$ on $D_1 = \{1,2\}$ and $\epsilon=0.5$ on $D_2 = \{2,4\}$. For both $D_1$ and $D_2$, there are only two clusterings $C_1 = \{(x1,x2)\}, C_2 = \{(x1),(x2)\}$. For $D_1$, $P_1(C_1) = \frac{e^{\epsilon score(C_1)}}{e^{\epsilon score(C_1)} + e^{\epsilon score(C_2)}} = \frac{e^{1.5}}{e^{1.5} + e^{2}}$.  For $D_2$, $P_2(C_1) = \frac{e^3}{e^3 + e^4}$. We can see $P_1(C_1) \neq P_2(C_2)$.
\end{example}
\end{proof}

\begin{theorem}\label{thm:DAWA-se}
DAWA ensures scale-epsilon exchangeability.
\end{theorem}
\begin{proof}
DAWA assumes order on the domain and perturbs the bias of all possible partitions by adding Laplace noise. Then it computes the partitioning based on the these noisy bias to minimize the objective function of $Error = \sum_{c \in P}bias(c) + \frac{|P|}{\epsilon}$ on any partitioning P. Let $D_1$ and $D_2$ be two datasets with the same shape but with scales m and $\alpha$m respectively. For $D_1$ under $\epsilon$, DAWA combines partitions $c_1$ and $c_2$ if $\tilde{bias}(c_1) + \tilde{bias}(c_2) + \frac{2}{\epsilon} > \tilde{bias}(c_1 + c_2) + \frac{1}{\epsilon}$. We have (suppose $bias(c_1)$ = x, $bias(c_2)$ = y, $bias(c_1+c_2)$ = z)
\begin{eqnarray*}
&&P(\tilde{bias}(c_1) + \tilde{bias}(c_2) + \frac{2}{\epsilon} > \tilde{bias}(c_1 + c_2) + \frac{1}{\epsilon} | D_1)\\
&=& P(Lap(x,\frac{1}{\epsilon} )+ Lap(y,\frac{1}{\epsilon}) + \frac{1}{\epsilon} > Lap(z,\frac{1}{\epsilon})) \\
&=& \int_a \int_b f(a | x,\frac{1}{\epsilon})f(b | y,\frac{1}{\epsilon}) F(a+b+\frac{1}{\epsilon} | z,\frac{1}{\epsilon}) dadb \\
&=& \int_a \int_b \alpha^2 f(\alpha a | \alpha x,\frac{\alpha}{\epsilon})f(\alpha b| \alpha y,\frac{\alpha}{\epsilon}) F(\alpha a+\alpha b+\frac{\alpha }{\epsilon} | \alpha z,\frac{\alpha}{\epsilon}) dadb \\
&=&\int_c \int_d f(c | \alpha x,\frac{\alpha}{\epsilon})f(d| \alpha y,\frac{\alpha}{\epsilon}) F(c+d+\frac{\alpha }{\epsilon} | \alpha z,\frac{\alpha}{\epsilon}) dcdd~(c = \alpha a, d = \alpha b)\\
&=& P(Lap(\alpha x,\frac{\alpha }{\epsilon}) + Lap(\alpha y,\frac{\alpha }{\epsilon}) + \frac{2\alpha }{\epsilon} > Lap(\alpha z,\frac{\alpha }{\epsilon}) + \frac{\alpha }{\epsilon})\\
&=&P(\tilde{bias}(c_1) + \tilde{bias}(c_2) + \frac{2\alpha}{\epsilon} > \tilde{bias}(c_1 + c_2) + \frac{\alpha}{\epsilon} | D_2)
\end{eqnarray*}
Since both $D_1$ under $\epsilon$ and $D_2$ under $\frac{\epsilon}{\alpha}$ combine $c_1$ and $c_2$ with the same probability, for any partitioning P of the domain, Pr[DAWA returns P with $D_1$ under $\epsilon$] = Pr[DAWA returns P with $D_2$ under $\frac{\epsilon}{\alpha}$]. 

In the second step, DAWA uses Laplace Mechanism to estimate the partition counts. Given a partitioning, the scaled error is the same in $D_1$ under $\epsilon$ and $D_2$ under $\frac{\epsilon}{\alpha}$. Since the probability a partitioning P is chosen is the same in both cases, the expected scaled error for DAWA is the same in both cases, which ensures scale-epsilon exchangeability.
\end{proof}

\begin{theorem}
AHP ensures scale-epsilon exchangeability.
\end{theorem}
\begin{proof}
The proof is very similar to that of \cref{thm:DAWA-se}.
\end{proof}

\begin{theorem}
\AG, \UG, \quadtree and \hybridtree are scale-epsilon exchangeable.
\end{theorem}
\begin{proof}
The partitioning or the tree generation of these four algorithms are independent with $m\cdot\epsilon$, where $m$ is the scale of the input dataset and $\epsilon$ is the privacy budget. After the partitioning or tree generation, each algorithm applies Laplace Mechanism to answer node queries.  Laplace Mechanism is scale-epsilon exchangeable. Thus, \AG, \UG, \quadtree and \hybridtree ensure scale-epsilon exchangeability.
\end{proof}

\end{document}